\documentclass[11pt,a4paper]{article}
\pdfoutput=1
\usepackage[utf8]{inputenc}
\usepackage{jheppub}
\hypersetup{unicode=true,bookmarksopen=true}

\usepackage{subcaption}

\usepackage{amsthm}
\usepackage{mathtools}
\usepackage{lmodern}
\usepackage[T1]{fontenc}
\usepackage{bbm}
\DeclareMathAlphabet{\mathbfi}{OML}{cmm}{b}{it}

\usepackage[nointegrals]{wasysym}

\let\originalleft\left
\let\originalright\right
\renewcommand{\left}{\mathopen{}\mathclose\bgroup\originalleft}
\renewcommand{\right}{\aftergroup\egroup\originalright}
\makeatletter
\newcommand{\biggg}{\bBigg@\thr@@}
\newcommand{\Biggg}{\bBigg@{3.5}}
\makeatother

\makeatletter
\newenvironment{equations}[1][]{\subequations\ifx\relax#1\relax\else\label{#1}\fi\align\ignorespaces}{\endalign\ignorespacesafterend\endsubequations}
\def\@spliteq#1{\begin{equation}\begin{split}#1\end{split}\end{equation}}
\def\@spliteqstar#1{\begin{equation*}\begin{split}#1\end{split}\end{equation*}}
\def\splitequation{\collect@body\@spliteq}
\expandafter\def\csname splitequation*\endcsname{\collect@body\@spliteqstar}

\expandafter\def\csname endsplitequation*\endcsname{\ignorespacesafterend}
\makeatother

\renewcommand{\vec}[1]{{\ifnum9<1#1\mathbf{#1}\else\ifcat\noexpand#1\relax\boldsymbol{#1}\else\mathbfi{#1}\fi\fi}}
\newcommand{\mathe}{\mathrm{e}}
\newcommand{\mathi}{\mathrm{i}}
\let\oldre\Re
\let\oldim\Im
\renewcommand{\Re}{\oldre\mathfrak{e}\,}
\renewcommand{\Im}{\oldim\mathfrak{m}\,}
\newcommand{\total}{\mathop{}\!\mathrm{d}}
\newcommand{\abs}[1]{{\left\lvert{#1}\right\rvert}}
\newcommand{\norm}[1]{{\left\lVert{#1}\right\rVert}}
\newcommand{\sgn}{\operatorname{sgn}}
\newcommand{\artanh}{\operatorname{artanh}}
\newcommand{\1}{\mathbbm{1}}

\newcommand{\eqend}[1]{\,#1}
\newcommand{\bigo}[1]{\mathcal{O}\left({#1}\right)}

\DeclareMathOperator*{\Res}{Res}
\newcommand{\expect}[1]{\left\langle{#1}\right\rangle}
\newcommand{\bessel}[3]{\mathop{}\!\mathrm{#1}_{#2}\left(#3\right)}
\newcommand{\hypergeom}[2]{\,{}_{#1}\mathrm{F}_{#2}}

\newcommand{\pvalue}{\operatorname{pv}}

\makeatletter
\DeclareRobustCommand*{\citeDLMFeq}{\hyper@normalise\citeDLMFeq@}
\def\citeDLMFeq@#1#2{\cite[\hyper@linkurl{Eq.~#1.#2}{https://dlmf.nist.gov/#1.E#2}]{DLMF}}
\makeatother

\allowdisplaybreaks

\makeatletter
\gdef\@fpheader{\strut}
\makeatother

\makeatletter
\g@addto@macro\@floatboxreset\centering
\makeatother

\newtheorem{lemma}{Lemma}

\begin{document}

\title{Modular Hamiltonian for fermions of small mass}

\author{Daniela Cadamuro,}
\author{Markus B. Fröb and}
\author{Christoph Minz}

\affiliation{Institut f{\"u}r Theoretische Physik, Universit{\"a}t Leipzig, Br{\"u}derstra{\ss}e 16, 04103 Leipzig, Germany}

\emailAdd{cadamuro@itp.uni-leipzig.de}
\emailAdd{mfroeb@itp.uni-leipzig.de}
\emailAdd{minz@itp.uni-leipzig.de}

\abstract{We consider the algebra of massive fermions restricted to a diamond in two-dimensional Minkowski spacetime, and in the Minkowski vacuum state. While the massless modular Hamiltonian is known for this setting, the derivation of the massive one is an open problem. We compute the small-mass corrections to the modular Hamiltonian in a perturbative approach, finding some terms which were previously overlooked. Our approach can in principle be extended to all orders in the mass, even though it becomes computationally challenging.}


\maketitle

\section{Introduction}

Modular theory is an important tool in the mathematical formulation of quantum field theory \cite{Takesaki:1970,Araki:1970,Fredenhagen:1984,Borchers:2000,Witten:2018,Longo:2022}. Especially in recent years, it gained further interest in applications in physics, because it serves as a method to compute relative entropies in quantum systems~\cite{CasiniHuerta:2009entropy,Hollands:2018,AbtErdmenger:2018,CasiniGrilloPontello:2019,Longo:2020,ErdmengerFriesReyesSimon:2020,DAngelo:2021,Bostelmann:2022,MorinelliTanimotoWegener:2022,CiolliLongoRanalloRuzzi:2022,GarbarzPalau:2022,HuertaVanDerVelde:2023,GalandaMuchVerch:2023,FroebMuchPapadopoulos:2023,DAngelo:2023}. For this purpose, one needs to know the modular Hamiltonian $H$ for the state and algebra under consideration. Unfortunately, explicit analytic expressions for modular Hamiltonians are only known in some cases, such as the algebra of quantum fields inside wedge regions in the Minkowski vacuum~\cite{BisognanoWichmann:1976}, or the algebra of free, massless scalar fields inside double cones or diamonds in the Minkowski vacuum~\cite{HislopLongo:1981}; see also~\cite{Froeb:2023} for conformal fields in de~Sitter spacetime and more references to earlier work.

For the special case of massless (Majorana) fermions in $(1+1)$-dimensional Minkowski spacetime, the situation is better, and many more results are known. In particular, the modular Hamiltonian has been determined for multi-component regions in Minkowski~\cite{CasiniHuerta:2009,Rehren:2010,Hollands:2021}, on a flat cylinder~\cite{KlichVamanWong:2017} and on a torus~\cite{BlancoPerezNadal:2019,FriesReyes:2019,BlancoGarbarzPerezNadal:2019}. However, results for massive fields have been obtained only numerically with lattice computations~\cite{AriasBlancoCasiniHuerta:2017,EislerDiGiulioTonniPeschel:2020,JaverzatTonni:2022} and other approaches~\cite{BostelmannCadamuroMinz:2023}. In this work, we want to obtain analytic results for massive fermions.

Namely, we consider the example of free, massive fermions inside a double cone or diamond of size $\ell$ in $(1+1)$-dimensional Minkowski spacetime. This region is the causal closure of the interval $V = [-\ell,\ell]$ on the Cauchy surface $t = 0$, on which initial data is given, see figure~\ref{fig:Diamond}. To obtain analytic results in the massive theory, we employ perturbation theory and compute the first-order corrections to the known massless result~\cite{CasiniHuerta:2009}. For free Majorana fermions\footnote{Since a Dirac fermion can be decomposed into two Majorana fermions, this restriction does not entail a loss of generality.}, it is possible to determine the modular Hamiltonian directly from the two-point function according to~\cite{Araki:1970,Peschel:2003,CasiniHuerta:2009}
\begin{equation}
\label{eq:ModularHamiltonian.Region}
H_V = - \ln\left( \mathcal{G}_V^{-1} - \1_V \right) = 2 \mathi \arctan\Bigl[ \mathi \bigl( \1_V - 2 \mathcal{G}_V \bigr) \Bigr] \eqend{.}
\end{equation}
This relation needs to be understood as an equality between convolution operators acting on initial data restricted to the interval $V$. Namely, it relates the integral kernel of the modular Hamiltonian $H_V$ with the integral kernel $\mathcal{G}_V$, which is the restriction of the Wightman two-point function
\begin{equation}
\label{eq:TwoPointFunction.Wightman}
\mathcal{G}(x,y) = \expect{ \psi(x) \psi^\dagger(y) } \eqend{,}
\end{equation}
evaluated on the Cauchy surface $t = 0$, to the interval $V$. It is known~\cite[Lemma~3.2]{Araki:1970} that the operator with integral kernel $\mathcal{G}_V$ is a bounded operator, whose spectrum is contained in the interval $[0,1]$. Moreover, $0$ and $1$ are not eigenvalues~\cite[Corollary~4.10]{Araki:1970}, and hence we can make sense of the formula~\eqref{eq:ModularHamiltonian.Region} using spectral theory.
\begin{figure}
\includegraphics{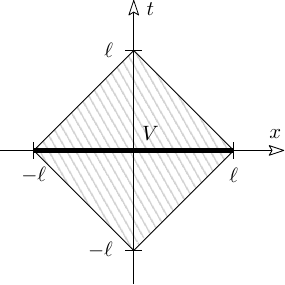}
\caption{The double cone of size $\ell$ in $(1+1)$-dimensional Minkowski spacetime, together with the interval $V = [-\ell,\ell]$ on the initial-data Cauchy surface.}
\label{fig:Diamond}
\end{figure}

For massless fermions, the spectral decomposition of $\mathcal{G}_V$ was determined in~\cite{CasiniHuerta:2009}, and the modular Hamiltonian computed using the formula~\eqref{eq:ModularHamiltonian.Region}. Building on these results, we treat the massive case in perturbation theory and compute explicitly the first-order corrections to the modular Hamiltonian. In principle, our approach can be used to compute the massive modular Hamiltonian to any order, even though the computations quickly become computationally difficult.

On the other hand, there exists a general formula for the modular Hamiltonian in the context of standard subspaces of Hilbert spaces~\cite{FiglioliniGuido:1989,FiglioliniGuido:1994,Longo:2022}. While formula~\eqref{eq:ModularHamiltonian.Region} only gives a result for the modular Hamiltonian restricted to the interval $V$, the general formula makes no such restriction. However, formula~\eqref{eq:ModularHamiltonian.Region} is much easier to use for practical computations. Obviously, it is important to know that the two formulas give the same result (when restricted to the interval $V$), and we show in section~\ref{sec:TomitaModularHamiltonian} that this indeed holds. In section~\ref{sec:FreeFermions} we compute the Wightman function for free massive fermions, briefly explain the spectral calculus for a purely continuous spectrum in subsection~\ref{sec:SpectralDecomposition}, and revisit the spectral decomposition of the massless Wightman function and the computation of the massless modular Hamiltonian in subsection~\ref{sec:ModularHamiltonian.Massless}. Section~\ref{sec:MassiveModularHamiltonian} is devoted to our main result, the computation of the massive modular Hamiltonian for small masses in perturbation theory. We conclude in section~\ref{sec:Discussion}, and leave the detailed and intricate calculations of various integrals that are needed for our result to the appendices~\ref{appx:MasslessDecomposition} and~\ref{appx:FirstOrderMassiveCorrection}.

\section{From the fermionic Tomita operator to the modular Hamiltonian}
\label{sec:TomitaModularHamiltonian}

We first show that for free fermions, the formula~\eqref{eq:ModularHamiltonian.Region} is equivalent to the general expression for the modular Hamiltonian $H_V$ that was obtained in~\cite{FiglioliniGuido:1989,FiglioliniGuido:1994,Longo:2022} for an arbitrary subspace region $V$, when restricted to the corresponding subspace.

Let us introduce the expression of the modular generator from~\cite{FiglioliniGuido:1989}. Since we consider the modular operator for a free field theory, the modular operator on Fock space is the second quantisation of the one-particle modular operator $\varDelta$~\cite{Foit:1983,FiglioliniGuido:1994}, hence it suffices to understand the one-particle structure. Consider thus the one-particle Hilbert space $\mathcal{H}$ with complex structure $I$, which is an operator acting on $\mathcal{H}$ that satisfies $I^2 = - \1$ and $I^\dagger = - I$.\footnote{While $\mathi \1$ is an operator that fulfills these conditions if $\mathcal{H}$ is a complex Hilbert space, in order for the subspace $\mathcal{L}$ to be standard a different complex structure is required. This structure is related to the two-point function as discussed below.} Consider furthermore a real-linear, standard subspace $\mathcal{L}$, which means that it is separating ($\mathcal{L} \cap I \mathcal{L} = \{ 0 \}$) and cyclic ($\overline{\mathcal{L} + I \mathcal{L}} = \mathcal{H}$). The Tomita operator $T$ for the subspace $\mathcal{L}$ is the closure of the map (for all $h, k \in \mathcal{L}$)
\begin{equation}
\label{eq:TomitaOperator}
h + I k \mapsto h - I k \eqend{.}
\end{equation}
Its polar decomposition $T = J \varDelta^{1/2}$ is given by the modular conjugation $J$ (an antiunitary involution) and the modular operator $\varDelta$ (a positive, non-singular, selfadjoint, linear operator).

The orthogonal projector $E$ on the standard subspace $\mathcal{L} = E \mathcal{H}$ (corresponding to the fields inside a region $V$) is expressed by the Tomita operator as follows~\cite{Longo:2022,FiglioliniGuido:1994}:
\begin{equation}
\label{eq:OrthogonalProjector}
E = ( \1 + T ) ( \1 + \varDelta )^{-1} \eqend{,}
\end{equation}
any operator $A$ can be projected to the subspace via $A_V = E A E$, and $E$ has norm 1. The complex structure commutes with the modular operator $\varDelta$, but anticommutes with the Tomita operator $T$, which follows immediately from the definition of the Tomita operator~\eqref{eq:TomitaOperator} and its polar decomposition. We can therefore isolate an expression for the modular generator by computing
\begin{equations}[eq:ModularGenerator]
\1 - E + I E I &= ( \varDelta - \1 ) ( \varDelta + \1 )^{-1} = \tanh\left( \frac{1}{2} \ln \varDelta \right) \eqend{,} \label{eq:ModularGenerator.Derivation} \\
\ln \varDelta &= 2 \artanh\left( \1 - E + I E I \right) \label{eq:ModularGenerator.LogDef} \eqend{.}
\end{equations}
Note that the operator $I \ln \varDelta$ leaves the subspace $\mathcal{L}$ invariant~\cite[Prop.~2.1]{LongoMorsella:2023}, and thus its restriction to the subspace acts only on the subspace, 
\begin{equation}
\label{eq:ModularGeneratorSubspace}
( I \ln \varDelta )_V = 2 E I \artanh\left( \1 - E + I E I \right) E \eqend{.}
\end{equation}

We now show that this is actually equal to the formula~\eqref{eq:ModularHamiltonian.Region}.
\begin{lemma}
\label{lemma:modular}
We have
\begin{equation}
\label{eq:ModularGeneratorSubspaceRestricted}
( I \ln \varDelta )_V = - 2 \arctan(E I E) \eqend{.}
\end{equation}
\end{lemma}
\begin{proof}
Consider the operator $R \coloneqq - I + I E + E I$, whose relation to the modular generator~\eqref{eq:ModularGenerator.Derivation} is given by
\begin{equation}
R = - I ( \varDelta - \1 ) ( \varDelta + \1 )^{-1} \eqend{.}
\end{equation}
Since $\varDelta$ is a positive operator and the spectrum of $I$ is given by $\{ \pm \mathi \}$, the spectrum of $R$ is contained in the interval $[-\mathi,\mathi]$. Moreover, since $\varDelta$ is invertible, $\pm \mathi$ are not eigenvalues of $R$. Therefore, the identity $\artanh(\mathi z) = \mathi \arctan(z)$ (valid for all $z \in \mathbb{C} \setminus \{ \pm \mathi \}$) holds on the spectrum of $R$, possibly apart from a subset of spectral measure zero. Using functional calculus, we may thus rewrite Eq.~\eqref{eq:ModularGenerator.LogDef} as
\begin{equation}
\ln \varDelta = 2 \artanh( I R ) = 2 I \arctan( R ) \eqend{,}
\end{equation}
from which it follows that
\begin{equation}
\label{eq:ModularGeneratorFromR}
I \ln \varDelta = - 2 \arctan( R ) \eqend{,} \quad ( I \ln \varDelta )_V = - 2 E \arctan( R ) E \eqend{.}
\end{equation}
The identity $R E = E I E = E R$ shows that $R$ and $E$ commute, and consequently $E$ also commutes with the spectral projections $P_r$ (with $r \in [0,1]$) of $R$ on the interval $[-\mathi r,\mathi r]$. Therefore, given a vector $h \in \mathcal{H}$ in the range of $P_r$ we have $E h = E P_r h = P_r E h$, and $E h$ also lies in the range of $P_r$. Taking $r < 1$, the series expansion of the arctangent applied to $h$ converges strongly:
\begin{equation}
\arctan( R ) h = \lim_{n \to \infty} \sum_{k=0}^n \frac{(-1)^k}{1+2k} R^{1+2k} h \eqend{,}
\end{equation}
with the $n$-independent bound $\norm{ \sum_{k=0}^n \frac{(-1)^k}{1+2k} R^{1+2k} h } \leq \artanh(r) \norm{ h } < \infty$, and the same holds for $E h$. In fact, we have
\begin{splitequation}
\norm{ \sum_{k=n+1}^\infty \frac{(-1)^k}{1+2k} R^{1+2k} h } &\leq \sum_{k=n+1}^\infty \frac{r^{1+2k}}{1+2k} \norm{ h } = \sum_{k=0}^\infty \frac{r^{2n+3+2k}}{2n+3+2k} \norm{ h } \\
&\leq \frac{r^{2n+3}}{2n+3} \sum_{k=0}^\infty r^{2k} \norm{h} = \frac{r^{2n+3}}{(2n+3) (1-r^2)} \norm{h} \eqend{,}
\end{splitequation}
which shows uniform convergence on $P_r \mathcal{H}$. Therefore, we may write
\begin{splitequation}
\arctan( R ) E h &= \lim_{n \to \infty} \sum_{k=0}^n \frac{(-1)^k}{1+2k} R^{1+2k} E h \\
&= \lim_{n \to \infty} \sum_{k=0}^n \frac{(-1)^k}{1+2k} ( E I E )^{1+2k} h = \arctan(E I E) h \eqend{,}
\end{splitequation}
where we used the previous identity $R E = E I E$ and that $E$ is a projection. Consequently, we have
\begin{equation}
( I \ln \varDelta )_V h = - 2 \arctan(E I E) h
\end{equation}
for all vectors $h \in \mathcal{H}$ in the range of $P_r$. Since the set of all such vectors (for all $r < 1$) is a core for $\arctan(R)$ and thus for $I \ln \varDelta$~\eqref{eq:ModularGeneratorFromR} as well as for the projections on the subspace $\mathcal{L}$, the equality~\eqref{eq:ModularGeneratorSubspaceRestricted} follows.
\end{proof}

We now consider $(1+1)$-dimensional Minkowski spacetime, and take the Hilbert space $\mathcal{H} = L^2_\mathbb{R}(\mathbb{R}) \oplus L^2_\mathbb{R}(\mathbb{R})$ of (real-valued) initial data on the Cauchy surface $t = 0$. The complex structure on $\mathcal{H}$ is determined by the Wightman function $\mathcal{G}$ for free Majorana fermions,
\begin{equation}
\label{eq:ComplexStructure}
I = - \mathi ( \1 - 2 \mathcal{G} )
\end{equation}
and is a real antilocal operator, which is shown in detail in appendix~\ref{sec:AppendixFreeFermions}. The subspace $\mathcal{L}$ corresponds to test functions that are supported inside the region $V$, and $E$ is the operator on $\mathcal{H}$ that multiplies initial data with the characteristic function of $V$.\footnote{To show that $\mathcal{L}$ is standard, one proceeds in analogy to the case of free bosons, see Refs.~\cite{FiglioliniGuido:1989,SegalGoodman:1965,Masuda:1972,Murata:1973}.} The modular generator on the subspace~\eqref{eq:ModularGeneratorSubspaceRestricted} is thus related to the modular Hamiltonian as
\begin{equation}
( I \ln \varDelta )_V = - 2 \arctan(E I E) = 2 \arctan\Bigl[ \mathi \bigl( \1_V - 2 \mathcal{G}_V \bigr) \Bigr] = - \mathi H_V \eqend{,}
\end{equation}
where we used formula~\eqref{eq:ModularHamiltonian.Region} in the last step. Note that while formula~\eqref{eq:ModularHamiltonian.Region} can also be employed for complex-valued initial data (i.e., Dirac fermions), we only consider real-valued initial data (i.e., Majorana fermions). Our derivation shows that the restriction to real-valued data commutes with the convolution by $- \mathi H_V$, such that one can consistently restrict $- \mathi H_V$ to Majorana fermions.

We therefore can use the formula~\eqref{eq:ModularHamiltonian.Region} to study the modular generator $I \ln \varDelta$ on the subspace corresponding to fields restricted to a given region $V$. As explained in the introduction, in the following calculations we consider an interval $V = [-\ell,\ell]$, whose causal closure is a double cone, see figure~\ref{fig:Diamond}.

\section{Free fermions on \texorpdfstring{$(1+1)$}{(1+1)}-dimensional Minkowski spacetime}
\label{sec:FreeFermions}

For the formulation of the fermion field, we use the conventions of~\cite{FreedmanVanProeyen:2012}. A Dirac fermion is a complex two-component spinor $\psi = (\psi_1,\psi_2)$, and we choose the $\gamma$ matrices in the form
\begin{equation}
\label{eq:DiracMatrices}
\gamma^0 = \begin{pmatrix} 0 & 1 \\ - 1 & 0 \end{pmatrix} \eqend{,} \quad \gamma^1 = \begin{pmatrix} 0 & 1 \\ 1 & 0 \end{pmatrix} \eqend{,} \quad \gamma_* = \gamma^0 \gamma^1 = \begin{pmatrix} 1 & 0 \\ 0 & -1 \end{pmatrix} \eqend{.}
\end{equation}
This choice has the advantage~\cite{FreedmanVanProeyen:2012} that Majorana fermions are simply real, $\psi_i^* = \psi_i$.

The computation of the Wightman two-point function~\eqref{eq:TwoPointFunction.Wightman} in the Minkowski vacuum is standard, and we have included it for completeness in appendix~\ref{sec:AppendixFreeFermions}. The result reads
\begin{equations}[eq:TwoPointFunction.Components]
\mathcal{G}_{11}(x,y) &= \frac{1}{2 \pi \mathi} \left( \lim_{\epsilon \to 0^+} \frac{1}{x-y - \mathi \epsilon} + \frac{F_1(\abs{x-y}) - 1}{x-y} \right) \eqend{,} \\
\mathcal{G}_{12}(x,y) &= \frac{1}{2 \pi \mathi} F_0(\abs{x-y}) \eqend{,} \\
\mathcal{G}_{21}(x,y) &= - \frac{1}{2 \pi \mathi} F_0(\abs{x-y}) \eqend{,} \\
\mathcal{G}_{22}(x,y) &= - \frac{1}{2 \pi \mathi} \left( \lim_{\epsilon \to 0^+} \frac{1}{x-y + \mathi \epsilon} + \frac{F_1(\abs{x-y}) - 1}{x-y} \right) \eqend{,}
\end{equations}
where we defined the functions
\begin{equations}[eq:TwoPointFunction.MassiveTerms]
F_0(z) &\coloneqq m \bessel{K}0{m z} = - m \ln\left( \frac{m z}{2} \mathe^{\gamma_\mathrm{E}} \right) + \bigo{ m^3 \ln m } \eqend{,} \\
F_1(z) &\coloneqq m z \bessel{K}1{m z} = 1 + \frac{1}{2} m^2 z^2 \left[ \ln\left( \frac{m z}{2} \mathe^{\gamma_\mathrm{E}} \right) - \frac{1}{2} \right] + \bigo{ m^4 \ln m } \eqend{,}
\end{equations}
in terms of the modified Bessel functions (of the second kind) $\bessel{K}i{x}$ and their expansion with the Euler--Mascheroni constant $\gamma_\mathrm{E}$. 
Restricting $x$ and $y$ to the interval $[-\ell,\ell]$ yields the integral kernel of the operator $\mathcal{G}_V$, which is a convolution operator acting on functions restricted to the interval.

\subsection{Spectral decomposition on the interval}
\label{sec:SpectralDecomposition}

We explain the general theory of spectral decomposition, which was also used in \cite{CasiniHuerta:2009} to derive the decomposition for a multi-component region and massless fermions. As mentioned in the introduction, the spectrum of $\mathcal{G}_V$ is contained in $[0,1]$~\cite{Araki:1970}, and we assume that it is purely continuous. That is, we assume that we have generalised eigenvectors $\varPsi^{(k)}_a(s,x)$ and real eigenvalues $\lambda^{(k)}(s)$ such that the eigenequation
\begin{equation}
\label{eq:TwoPointFunction.EigenvectorEquation}
\sum_b \int_{-\ell}^\ell \mathcal{G}_{ab}(x,y) \varPsi^{(k)}_b(s,y) \total y = \lambda^{(k)}(s) \varPsi^{(k)}_a(s,x)
\end{equation}
holds. It turns out that it is useful to parametrise the eigenvalues and corresponding (generalised) eigenvectors by $s \in \mathbb{R}$, instead of taking the eigenvalue itself. The eigenvectors should be orthogonal, normalised (to the $\delta$ distribution since they are generalised eigenvectors), and complete:
\begin{equations}[eq:Eigenvector.Properties]
\sum_a \int_{-\ell}^\ell \varPsi^{(k)*}_a(s,x) \varPsi^{(l)}_a(t,x) \total x &= \delta_{kl} \delta(s-t) \eqend{,} \label{eq:Eigenvector.Properties.on} \\
\sum_k \int_{-\infty}^\infty \varPsi^{(k)*}_a(s,x) \varPsi^{(k)}_b(s,y) \total s &= \delta_{ab} \delta(x-y) \eqend{.} \label{eq:Eigenvector.Properties.complete}
\end{equations}
Then we obtain the spectral decomposition of the operator $\mathcal{G}_V$,
\begin{equation}
\label{eq:TwoPointFunction.SpectralDecomposition}
\mathcal{G}_{ab}(x,y) = \sum_k \int_{-\infty}^\infty \lambda^{(k)}(s) \varPsi^{(k)}_a(s,x) \varPsi^{(k)*}_b(s,y) \total s \eqend{,}
\end{equation}
and the resolvent 
\begin{equation}
\mathcal{R}(\mu) \coloneqq \left( \mathcal{G}_V - \mu \1 \right)^{-1} \eqend{,}
\end{equation}
exists for any complex value $\mu$ that is not in the spectrum of $\mathcal{G}_V$. By spectral calculus, the integral kernel of the resolvent reads 
\begin{equation}
\label{eq:Resolvent.Massless}
\mathcal{R}(\mu)_{ab}(x,y) = \sum_k \int_{-\infty}^\infty \frac{1}{\lambda^{(k)}(s) - \mu} \varPsi^{(k)}_a(s,x) \varPsi^{(k)*}_b(s,y) \total s \eqend{,}
\end{equation}
and the kernel of the modular Hamiltonian~\eqref{eq:ModularHamiltonian.Region} is given by
\begin{equation}
\label{eq:ModularHamiltonian.Massless}
H_{ab}(x,y) = - \sum_k \int_{-\infty}^\infty \ln\left( \frac{1}{\lambda^{(k)}(s)} - 1 \right) \varPsi^{(k)}_a(s,x) \varPsi^{(k)*}_b(s,y) \total s \eqend{.}
\end{equation}
This expression can be explicitly evaluated for the case of a massless fermion, which we review first. Afterwards, we use the resolvent to compute massive corrections in section~\ref{sec:MassiveContributions.FirstOrder}.

\subsection{The massless modular Hamiltonian}
\label{sec:ModularHamiltonian.Massless}

As a concrete example, let us consider the massless case $m = 0$, where the integral kernels of the components of the two-point function~\eqref{eq:TwoPointFunction.Components} read
\begin{equations}[eq:TwoPointFunction.Components.Massless]
\mathcal{G}_{11}(x,y) &= \lim_{\epsilon \to 0^+} \frac{1}{2 \pi \mathi} \frac{1}{x-y - \mathi \epsilon} \eqend{,} \\
\mathcal{G}_{12}(x,y) &= 0 \eqend{,} \\
\mathcal{G}_{21}(x,y) &= 0 \eqend{,} \\
\mathcal{G}_{22}(x,y) &= - \lim_{\epsilon \to 0^+} \frac{1}{2 \pi \mathi} \frac{1}{x-y + \mathi \epsilon} \eqend{.}
\end{equations}
We recall that these kernels define a convolution operator on the subspace, where $x,y \in [-\ell,\ell]$, and we always assume this restriction on $x$ and $y$ in the following. The generalised eigenvectors are given by~\cite{CasiniHuerta:2009}
\begin{splitequation}
\label{eq:MasslessEigenvector}
\varPsi^{(k)}_a(s,x) &= \delta^k_a \varPsi(s,x) \eqend{,} \\
\varPsi(s,x) &\coloneqq \sqrt{ \frac{\ell}{\pi} } (\ell+x)^{-\frac{1}{2}-\mathi s} (\ell-x)^{-\frac{1}{2}+\mathi s} = \sqrt{ \frac{\ell}{\pi} } \frac{1}{\sqrt{\ell^2 - x^2}} \left( \frac{\ell - x}{\ell + x} \right)^{\mathi s} \eqend{,}
\end{splitequation}
and the corresponding eigenvalues read
\begin{equation}
\label{eq:MasslessEigenvalue.Properties}
\lambda^{(1)}(s) = \frac{1}{1 + \mathe^{- 2 \pi s}} \eqend{,} \quad \lambda^{(2)}(s) = \frac{1}{1 + \mathe^{2 \pi s}} = \lambda^{(1)}(-s) \eqend{.}
\end{equation}
The function $\varPsi$ satisfies the useful properties
\begin{equation}
\label{eq:MasslessEigenvector.Properties}
\varPsi^*(s,x) = \varPsi(-s,x) = \varPsi(s,-x) \eqend{.}
\end{equation}
In appendix~\ref{appx:MasslessDecomposition}, we verify that these are indeed the correct generalised eigenvectors and eigenvalues, and that they form an orthogonal and complete eigenbasis, i.e., they satisfy the conditions~\eqref{eq:Eigenvector.Properties}. The distributional character of the orthonormality condition~\eqref{eq:Eigenvector.Properties.on} shows clearly that the generalised eigenvectors are not normalisable as elements of the Hilbert space. However, one can define them rigorously in the framework of rigged Hilbert spaces or Gel'fand triples~\cite{EncyclopediaMathRHS}.

Since $s \in \mathbb{R} = (-\infty,\infty)$, we obtain $\lambda^{(k)}(s) \in (0,1)$, verifying the condition on the spectrum of $\mathcal{G}_V$ and excluding the values $0$ and $1$. Note that the above basis diagonalizes the operator $\mathcal{G}_V$, as can be either inferred from the spectral decomposition of the propagator (see below), or from the results of Koppelman and Pincus~\cite{KoppelmanPincus:1959}, who show that the finite Hilbert transform (which $\mathcal{G}_V(x,y)$ essentially is) is unitarily equivalent to a multiplication operator (with some changes in notation and rescalings of variables).

The spectral decomposition of the propagator itself reads 
\begin{splitequation}
\mathcal{G}_{ab}(x,y) &= \sum_k \int_{-\infty}^\infty \lambda^{(k)}(s) \varPsi^{(k)}_a(s,x) \varPsi^{(k)*}_b(s,y) \total s \\
&= \delta_{ab} \frac{\ell}{\pi} \int_{-\infty}^\infty \frac{\lambda^{(a)}(s)}{\sqrt{( \ell^2 - x^2 ) ( \ell^2 - y^2 )}} \left( \frac{\ell - x}{\ell + x} \frac{\ell + y}{\ell - y} \right)^{\mathi s} \total s \eqend{,}
\end{splitequation}
and thus,
\begin{equation}
\label{eq:TwoPointFunction.Component11.SpectralDecomposition}
\mathcal{G}_{11}(x,y) = \frac{\ell}{\pi} \int_{-\infty}^\infty \frac{1}{1 + \mathe^{- 2 \pi s}} \frac{1}{\sqrt{( \ell^2 - x^2 ) ( \ell^2 - y^2 )}} \left( \frac{\ell - x}{\ell + x} \frac{\ell + y}{\ell - y} \right)^{\mathi s} \total s \eqend{,}
\end{equation}
$\mathcal{G}_{12} = \mathcal{G}_{21} = 0$ and $\mathcal{G}_{22}(x,y) = \mathcal{G}_{11}^*(x,y)$. Since we have generalised eigenvectors, the integral converges only in a distributional sense and is computed in appendix~\ref{appx:MasslessDecomposition}. In that appendix, we also show that the modular Hamiltonian, given by the spectral decomposition~\eqref{eq:ModularHamiltonian.Massless}, reads
\begin{splitequation}
\label{eq:ModularHamiltonian.Interval.Massless}
H_{ab}(x,y) &= 2 \pi \left( \delta^1_a \delta^1_b - \delta^2_a \delta^2_b \right) \int_{-\infty}^\infty s \varPsi(s,x) \varPsi(s,-y) \total s \\
&= 2 \ell \left( \delta^1_a \delta^1_b - \delta^2_a \delta^2_b \right) \int_{-\infty}^\infty \frac{s}{\sqrt{( \ell^2 - x^2 ) ( \ell^2 - y^2 )}} \left( \frac{\ell - x}{\ell + x} \frac{\ell + y}{\ell - y} \right)^{\mathi s} \total s \\
&= \mathi \pi \gamma_{*ab} \frac{\ell^2 - x y}{\ell} \delta'(x-y) \eqend{.}
\end{splitequation}
This coincides with the known results for a single interval~\cite{CasiniHuerta:2009}.

Similarly, we could also compute an explicit expression for the resolvent~\eqref{eq:Resolvent.Massless}. However, in the following calculations we only need its definition to compute massive corrections to the modular Hamiltonian, and thus refrain from giving the corresponding result.

\section{The massive modular Hamiltonian}
\label{sec:MassiveModularHamiltonian}

In the previous section, we computed the integral kernel of the massless modular Hamiltonian for the interval $[-\ell,\ell]$ and obtained the known result~\eqref{eq:ModularHamiltonian.Interval.Massless} starting with the massless part of the two-point function~\eqref{eq:TwoPointFunction.Components.Massless}. Now we use the spectral decomposition of the massless case to treat the massive modular Hamiltonian perturbatively for a small mass $m$. The perturbation operator, which we call $\mathcal{K}$, is the main tool in the computation of the first-order mass corrections.

\subsection{The perturbation operator}
\label{sec:PerturbationOperator}

Equation~\eqref{eq:TwoPointFunction.Components} shows that the massive contributions to the full two-point function are determined by modified Bessel functions, which have a series expansion for small values of the mass parameter $m$ given by~\eqref{eq:TwoPointFunction.MassiveTerms}. It is these terms that determine the integral kernels of the perturbation operator $\mathcal{K}$ in the lowest order $m \ln m$ of the massive correction:
\begin{equations}[eq:PerturbationOperator]
\mathcal{K}_{11}(x,y) &= \bigo{ m^2 \ln m } \eqend{,} \\
\mathcal{K}_{12}(x,y) &= - \frac{m}{2 \pi \mathi} \ln\left( \frac{m \abs{x-y}}{2} \mathe^{\gamma_\mathrm{E}} \right) + \bigo{ m^3 \ln m } \eqend{,} \\
\mathcal{K}_{21}(x,y) &= \frac{m}{2 \pi \mathi} \ln\left( \frac{m \abs{x-y}}{2} \mathe^{\gamma_\mathrm{E}} \right) + \bigo{ m^3 \ln m } \eqend{,} \\
\mathcal{K}_{22}(x,y) &= \bigo{ m^2 \ln m } \eqend{.}
\end{equations}
Since we are only interested in the first-order perturbation, we drop all terms of order $m^2 \ln m$ and higher. Calculations of higher-order corrections would follow essentially the same steps.

To compute the correction to the modular Hamiltonian via the resolvent, we first express the perturbation kernel $\mathcal{K}$ in the massless (generalised) eigenbasis given by the $\varPsi^{(k)}(s)$. The corresponding matrix elements are 
\begin{splitequation}
\label{eq:PerturbationOperator.MatrixElement}
K^{(kl)}(s,t) &\coloneqq \left( \varPsi^{(k)}(s), \mathcal{K} \varPsi^{(l)}(t) \right) \\
&= \sum_{a,b} \iint_{-\ell}^\ell \varPsi^{(k)*}_a(s,x) \mathcal{K}_{ab}(x,y) \varPsi^{(l)}_b(t,y) \total x \total y \\
&= m \ell \epsilon_{kl} K(s,t) \eqend{,} \\
K(s,t) &\coloneqq \frac{\mathi}{2 \pi \ell} \iint_{-\ell}^\ell \varPsi^*(s,x) \ln\left( \frac{m \abs{x-y}}{2} \mathe^{\gamma_\mathrm{E}} \right) \varPsi(t,y) \total x \total y \eqend{,}
\end{splitequation}
where $\epsilon_{12} = 1$, $\epsilon_{21} = -1$ and $\epsilon_{11} = \epsilon_{22} = 0$. The detailed calculations of the double integral $K(s,t)$ are deferred to appendix~\ref{appx:FirstOrderMassiveCorrection}. The solution reads
\begin{splitequation}
\label{eq:PerturbationOperator.MatrixElement.Solution}
K(s,t) &= \frac{\mathi \ln\left( m \ell \mathe^{2\gamma_\mathrm{E}} \right)}{2 \cosh(\pi s) \cosh(\pi t)} \\
&\quad- \frac{\mathi \tanh(\pi t)}{4 \sinh[ \pi (s-t) ]} \left[ \psi\left( \frac{1}{2} + \mathi s \right) + \psi\left( \frac{1}{2} - \mathi s \right) \right] \\
&\quad+ \frac{\mathi \tanh(\pi s)}{4 \sinh[ \pi (s-t) ]} \left[ \psi\left( \frac{1}{2} + \mathi t \right) + \psi\left( \frac{1}{2} - \mathi t \right) \right] \eqend{,}
\end{splitequation}
where $\psi(x) = \partial_x \ln \Gamma(x)$ is the digamma function, and we see that $K$ is a bounded (and actually fast decaying), smooth kernel with the properties 
\begin{equation}
\label{eq:PerturbationOperator.MatrixElement.Solution.Properties}
K(s,t) = K(t,s) = K(-s,-t) \eqend{.}
\end{equation}
It is easy to see that $K$ is smooth for all $s \neq t$, while in the coincidence case $t \to s$ we use l'H{\^o}pital's rule and obtain
\begin{splitequation}
K(s,s) &= \frac{\mathi}{4 \cosh^2(\pi s)} \left[ \psi\left( \frac{1}{2} + \mathi s \right) + \psi\left( \frac{1}{2} - \mathi s \right) + 2 \ln\left( m \ell \mathe^{2\gamma_\mathrm{E}} \right) \right] \\
&\quad- \frac{1}{4 \pi} \tanh(\pi s) \left[ \psi'\left( \frac{1}{2} - \mathi s \right) - \psi'\left( \frac{1}{2} + \mathi s \right) \right] \eqend{.}
\end{splitequation}

Given this solution for the matrix elements of the perturbation operator, we recover the operator kernel as the inverse relation to~\eqref{eq:PerturbationOperator.MatrixElement}. This means that we have the expansion 
\begin{equation}
\mathcal{K}_{ab}(x,y) = \sum_{k,l} \iint \varPsi^{(k)}_a(s,x) K^{(kl)}(s,t) \varPsi^{(l)*}_b(t,y) \total s \total t \eqend{,}
\end{equation}
which we use for the computation of the first-order perturbation to the massive Hamiltonian in the following.

\subsection{Massive first-order contributions}
\label{sec:MassiveContributions.FirstOrder}

The expression~\eqref{eq:ModularHamiltonian.Region} for the modular Hamiltonian is a special case of an integral identity. Consider an operator $A$ with spectrum contained in $[0,1]$, where $0$ and $1$ are not eigenvalues. Via spectral calculus, it fulfills the integral relation
\begin{equation}
\ln\left( A^{-1} - \1 \right) = \int_0^\infty \left[ ( A + \mu \1 )^{-1} - ( \1 - A + \mu \1 )^{-1} \right] \total \mu \eqend{.}
\end{equation}
When perturbing the operator $A \to A + \delta A$, we obtain 
\begin{splitequation}
\label{eq:PerturbationLogA}
\delta \ln\left( A^{-1} - \1 \right) &= - \int_0^\infty \Bigl[ ( A + \mu \1 )^{-1} \delta A ( A + \mu \1 )^{-1} \\
&\qquad\qquad+ ( \1 - A + \mu \1 )^{-1} \delta A ( \1 - A + \mu \1 )^{-1} \Bigr] \total \mu \eqend{.}
\end{splitequation}
To apply this to the modular Hamiltonian, we set $A = \mathcal{G}_V$ and note that the inverse operators of the integrand on the right-hand side correspond to the resolvent of the massless theory $\mathcal{R}(\mu) = \left( \mathcal{G}_V - \mu \1 \right)^{-1}$ for different values $\mu$. It follows that the first-order perturbation of the modular Hamiltonian reads 
\begin{equation}
H^{(1)} = \int_0^\infty \Bigl[ \mathcal{R}(-\mu) \, \mathcal{K} \, \mathcal{R}(-\mu) + \mathcal{R}(1+\mu) \, \mathcal{K} \, \mathcal{R}(1+\mu) \Bigr] \total \mu \eqend{.}
\end{equation}

The resolvent of the massless theory has the spectral expansion~\eqref{eq:Resolvent.Massless}. For the computation of the convolution $\mathcal{R}(\mu) \, \mathcal{K} \, \mathcal{R}(\mu)$, we use the orthonormality and completeness condition~\eqref{eq:Eigenvector.Properties} of the massless eigenbasis. This results in
\begin{splitequation}
\bigl[ \mathcal{R}(\mu) \, \mathcal{K} \, \mathcal{R}(\mu) \bigr]_{ab}(x,y)
&= \sum_{c,d} \iint_{-\ell}^\ell \mathcal{R}(\mu)_{ac}(x,u) \, \mathcal{K}_{cd}(u,v) \, \mathcal{R}(\mu)_{db}(v,y) \total u \total v \\
&= \sum_{k,l} \iint \frac{\varPsi^{(k)}_a(s,x)}{\lambda^{(k)}(s) - \mu} K^{(kl)}(s,t) \frac{\varPsi^{(l)*}_b(t,y)}{\lambda^{(l)}(t) - \mu} \total s \total t \eqend{.}
\end{splitequation}
We can now insert this formula into the first-order perturbation of the modular Hamiltonian and integrate over $\mu$, which results in
\begin{splitequation}
H^{(1)}_{ab}(x,y) &= \sum_{k,l} \iint \varPsi^{(k)}_a(s,x) K^{(kl)}(s,t) \varPsi^{(l)*}_b(t,y) \\
&\quad\times \int_0^\infty \left( \frac{1}{\lambda^{(k)}(s) + \mu} \frac{1}{\lambda^{(l)}(t) + \mu}  + \frac{1}{\lambda^{(k)}(s) - \mu-1} \frac{1}{\lambda^{(l)}(t) - \mu-1} \right) \total \mu \total s \total t \\
&= \iint \frac{\varPsi(s,x) K^{(ab)}(s,t) \varPsi(-t,y)}{\lambda^{(a)}(s) - \lambda^{(b)}(t)} \ln\left( \frac{\lambda^{(a)}(s)}{\lambda^{(b)}(t)} \frac{1 - \lambda^{(b)}(t)}{1 - \lambda^{(a)}(s)} \right) \total s \total t \eqend{.}
\end{splitequation}
We use the symmetry properties~\eqref{eq:MasslessEigenvalue.Properties} and~\eqref{eq:PerturbationOperator.MatrixElement.Solution.Properties} for the eigenvalues and the kernel $K(s,t)$, respectively, to simplify this expression and obtain
\begin{equations}
H^{(1)}_{11}(x,y) &= H^{(1)}_{22}(x,y) = 0 \eqend{,} \\
H^{(1)}_{12}(x,y) &= 4 \pi m \ell \iint \varPsi(s,x) K(s,-t) \varPsi(t,y) \frac{(s-t) \cosh(\pi s) \cosh(\pi t)}{\sinh[ \pi (s-t) ]} \total s \total t \eqend{,} \\
H^{(1)}_{21}(x,y) &= \left[ H^{(1)}_{12}(y,x) \right]^* \eqend{.}
\end{equations}
With the previously computed matrix element~\eqref{eq:PerturbationOperator.MatrixElement.Solution}, the remaining integral is 
\begin{splitequation}
\label{eq:ModularHamiltonian.Interval.Massive.Integral}
H^{(1)}_{12}(x,y) &= \frac{\mathi m \ell^2}{\sqrt{ (\ell^2-x^2) (\ell^2-y^2) }} \iint \left( \frac{\ell-x}{\ell+x} \right)^{\mathi s} \left( \frac{\ell-y}{\ell+y} \right)^{\mathi t} \frac{s-t}{\sinh[ \pi (s-t) ]} \\
&\qquad\times \Bigg( 2 \ln\left( m \ell \mathe^{2\gamma_\mathrm{E}} \right) + \frac{\cosh(\pi s) \sinh(\pi t)}{\sinh[ \pi (s+t) ]} \left[ \psi\left( \frac{1}{2} + \mathi s \right) + \psi\left( \frac{1}{2} - \mathi s \right) \right] \\
&\qquad\qquad+ \frac{\sinh(\pi s) \cosh(\pi t)}{\sinh[ \pi (s+t) ]} \left[ \psi\left( \frac{1}{2} + \mathi t \right) + \psi\left( \frac{1}{2} - \mathi t \right) \right] \Bigg) \total s \total t \eqend{.}
\end{splitequation}
The computation of this double integral is somewhat involved, and we perform it in appendix~\ref{appx:FirstOrderMassiveCorrection}.

To write down the solution, let us define the distribution
\begin{equation}
\label{eq:ModifiedPrincipalValueDistribution}
\pvalue_\mu \frac{1}{\abs{x}} \coloneqq \lim_{\epsilon \to 0^+} \left[ \frac{\Theta(x-\epsilon) - \Theta(-x-\epsilon)}{x} + 2 \ln\left( \mu \epsilon \, \mathe^{-\gamma_\mathrm{E}} \right) \delta(x) \right] \eqend{,}
\end{equation}
which depends on a parameter $\mu$ with dimensions of mass to make the argument of the logarithm dimensionless. For convenience, one could set $\mu = \ell^{-1}$. In appendix~\ref{appx:FinitePartDistribution}, we also show that $\pvalue_\mu \frac{1}{\abs{x}}$ is a well-defined distribution, i.e., that the limit $\epsilon \to 0^+$ is finite after smearing with a test function. With this distribution, the solution of~\eqref{eq:ModularHamiltonian.Interval.Massive.Integral} reads
\begin{splitequation}
\label{eq:FirstOrderMassCorrection.Result}
H^{(1)}_{12}(x,y) &= 2 \pi \mathi m \ell \Biggl[ \ln\left( m \ell \frac{\ell^2 - x^2}{2 \ell} \mu \right) \frac{\ell^2 - x^2}{2 \ell^2} \delta(x+y) + \frac{1}{8 \ell^2} \abs{x-y} \\
&\qquad\qquad- \frac{\ell^2 - x^2}{2 \ell^2} \delta(x-y) - \frac{2 \ell^2 - x^2 - y^2}{8 \ell^2} \pvalue_\mu \frac{1}{\abs{x+y}} \Biggr] \eqend{.}
\end{splitequation}

Combing the massless result~\eqref{eq:ModularHamiltonian.Interval.Massless} and the first-order massive corrections~\eqref{eq:FirstOrderMassCorrection.Result}, we obtain the integral kernels
\begin{equations}[eq:ModularHamiltonian.Result]
H_{11}(x,y) &= \mathi \pi \frac{\ell^2 - x y}{\ell} \delta'(x-y) + \bigo{ m^2 \ln m } \eqend{,} \\
H_{12}(x,y) &= H^{(1)}_{12}(x,y) + \bigo{ m^2 \ln m } \eqend{,} \\
H_{21}(x,y) &= - H^{(1)}_{12}(x,y) + \bigo{ m^2 \ln m } \eqend{,} \\
H_{22}(x,y) &= - \mathi \pi \frac{\ell^2 - x y}{\ell} \delta'(x-y) + \bigo{ m^2 \ln m }
\end{equations}
of the full modular Hamiltonian for fermions in the interval $[-\ell,\ell]$ and to first order in the mass $m$. In the discussion below, we denote the massless contribution by $H^{(0)}_{11}(x,y)$.

\section{Discussion}
\label{sec:Discussion}

We have computed the full analytic form of the first-order massive corrections to the massless modular Hamiltonian, for fermions in $(1+1)$-dimensional Minkowski spacetime and the Minkowski vacuum state restricted to an interval. To understand these corrections, it is useful to have a visual representation of the integral kernels. Since they are distributions and not just functions, we have to smear them with test functions, which we take to be normalised Gaussians with centre $x_i$ and variance $\sigma^2$:
\begin{equation}
\label{eq:GaussianTestFunction}
g_i(x) \coloneqq \frac{1}{\sqrt[4]{\pi \sigma^2}} \exp\left( - \frac{( x - x_i )^2}{2 \sigma^2} \right) \eqend{,} \quad \int \left[ g_i(x) \right]^2 \total x = 1 \eqend{.}
\end{equation}
For the centres $x_i$ we choose equidistant points in the interval $[-\ell,\ell]$, and to obtain a high resolution we set the width to $\sigma = \ell/32$.

The massless contribution to the $11$ component of the modular Hamiltonian $\mathi H^{(0)}_{11}(x,y)$ as given in~\eqref{eq:ModularHamiltonian.Result} smeared against two sets of Gaussians in the $x$ and $y$ direction is shown in figure~\ref{fig:ModularHamiltonian.Massless}. 
\begin{figure}
    \includegraphics{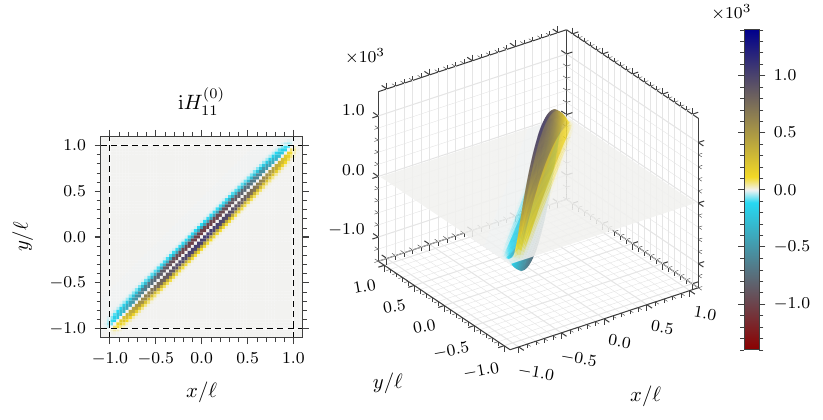}
    \caption{Massless integral kernel $\mathi H^{(0)}_{11}(x,y)$ smeared against Gaussian test functions $g_i(x) g_j(y)$ as given in~\eqref{eq:GaussianTestFunction}, shown once as a matrix plot (left) and a perspective on a surface plot (right) using the same colour scale. Note that this scale is many orders of magnitude larger than the first-order corrections shown in figure~\ref{fig:ModularHamiltonian.MassiveFirstOrder}.}
    \label{fig:ModularHamiltonian.Massless}
\end{figure}
For the massive corrections $\mathi H^{(1)}_{12}(x,y)$, we choose a small mass $m = 0.02/\ell$ (for which a first-order approximation should be valid) and smear against the same Gaussians, see figure~\ref{fig:ModularHamiltonian.MassiveFirstOrder}.
\begin{figure}[ht]
    \begin{subfigure}[b]{0.45\textwidth}
         \centering
         \includegraphics{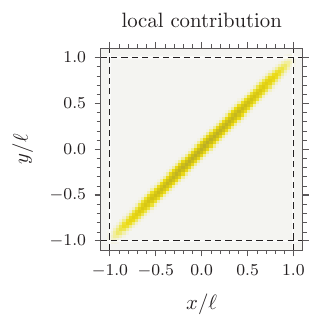}
         \caption{Terms proportional to $\delta(x-y)$.}
         \label{fig:ModularHamiltonian.MassiveFirstOrder.Local}
    \end{subfigure}
    \hspace{2em}
    \begin{subfigure}[b]{0.45\textwidth}
        \centering
        \includegraphics{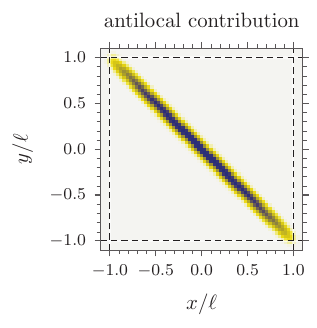}
        \caption{Terms proportional to $\delta(x+y)$.}
        \label{fig:ModularHamiltonian.MassiveFirstOrder.Antilocal}
    \end{subfigure}
    \\
    \begin{subfigure}[b]{0.45\textwidth}
        \centering
        \includegraphics{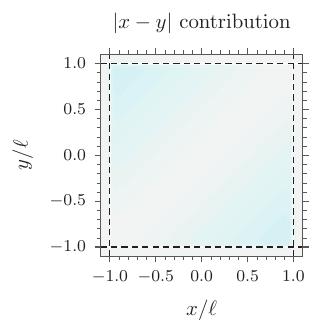}
        \caption{Terms proportional to $\abs{x-y}$.}
        \label{fig:ModularHamiltonian.MassiveFirstOrder.Abs}
    \end{subfigure}
    \hspace{2em}
    \begin{subfigure}[b]{0.45\textwidth}
        \centering
        \includegraphics{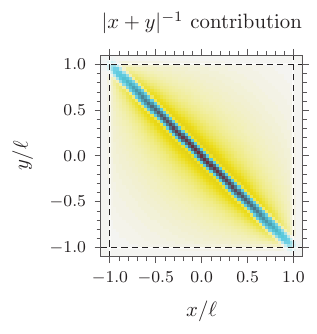}
        \caption{Terms proportional to $\frac{1}{\abs{x+y}}$.}
        \label{fig:ModularHamiltonian.MassiveFirstOrder.AbsInv}
    \end{subfigure}
    \\
    \begin{subfigure}[b]{0.95\textwidth}
        \centering
        \includegraphics{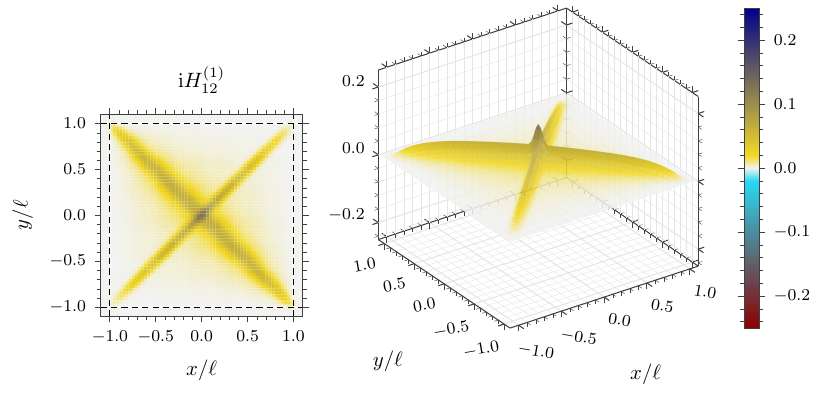}
        \caption{Sum of all four contributions figure~\ref{fig:ModularHamiltonian.MassiveFirstOrder.Local}, figure~\ref{fig:ModularHamiltonian.MassiveFirstOrder.Antilocal}, figure~\ref{fig:ModularHamiltonian.MassiveFirstOrder.Abs} and figure~\ref{fig:ModularHamiltonian.MassiveFirstOrder.AbsInv}.}
        \label{fig:ModularHamiltonian.MassiveFirstOrder.All}
    \end{subfigure}
    \caption{Corrections to the modular Hamiltonian in first order for a mass $m = 0.02/\ell$, smeared against Gaussians equidistantly shifted in $x$- and $y$-directions.}
    \label{fig:ModularHamiltonian.MassiveFirstOrder}
\end{figure}
For a more detailed inspection, we break the result~\eqref{eq:FirstOrderMassCorrection.Result} down into its four contributions. The first contribution is the local correction of the diagonal part that is proportional to $\delta(x-y)$, see figure~\ref{fig:ModularHamiltonian.MassiveFirstOrder.Local}. An even stronger contribution comes from the antilocal terms proportional to $\delta(x+y)$, see figure~\ref{fig:ModularHamiltonian.MassiveFirstOrder.Antilocal}. The absolute function $\abs{x-y}$, see figure~\ref{fig:ModularHamiltonian.MassiveFirstOrder.Abs}, contributes only mildly, while the regularised absolute inverse term $\pvalue_\mu \frac{1}{\abs{x+y}}$ (shown in figure~\ref{fig:ModularHamiltonian.MassiveFirstOrder.AbsInv}) partially compensates the antilocal term. In summary, figure~\ref{fig:ModularHamiltonian.MassiveFirstOrder.All} shows the full integral kernel integrated against the Gaussian functions as a matrix plot and also as a surface plot. All plots for the massive corrections use the same colour scale, while the massless contribution (see figure~\ref{fig:ModularHamiltonian.Massless}) is many orders of magnitude stronger, essentially because its distributional kernel is proportional to the derivative $\delta'(x-y)$. Incidentally, this shows that including the first-order corrections should be a good approximation to the full massive modular Hamiltonian for mass $m = 0.02/\ell$.

We note that some terms of the first-order massive correction for the modular Hamiltonian have been previously computed by Arias, Blanco, Casini and Huerta~\cite[Appx.~A.2]{AriasBlancoCasiniHuerta:2017}. Compared to our full result, one singular contribution is missing, and the nonsingular terms were not determined. The difference between the singular terms of our result and the one given in \cite[Appx.~A.2]{AriasBlancoCasiniHuerta:2017} is a logarithmic antilocal term 
\begin{equation}
\left[ H^{(1)}_{12}(x,y) \right]_\text{sing.} - \left[ H^{(1)}_{12}(x,y) \right]^{\text{\cite{AriasBlancoCasiniHuerta:2017}}} = \mathi \pi m \frac{\ell^2 - x^2}{\ell} \ln\left( \frac{\ell^2 - x^2}{\ell^2} \mathe^{- 2 \gamma_\mathrm{E}} \right) \delta(x+y) \eqend{,}
\end{equation}
which is included in the contribution shown in figure~\ref{fig:ModularHamiltonian.MassiveFirstOrder.Antilocal}. 
As explicit result for the small nonsingular contribution, we found 
\begin{equation}
\left[ H^{(1)}_{12}(x,y) \right]_\text{nonsing.} = \mathi \pi m \frac{\abs{x-y}}{4 \ell} \eqend{,}
\end{equation}
see also figure~\ref{fig:ModularHamiltonian.MassiveFirstOrder.Abs}. 

In principle, the perturbative calculation that we demonstrated for the first-order corrections can be continued to higher orders. As indicated in \eqref{eq:ModularHamiltonian.Result}, starting from the second order, further corrections change all spinor components of the modular Hamiltonian making the computations increasingly involved. Moreover, one needs both to compute higher-order corrections of the perturbation operator $\mathcal{K}$ as defined in~\eqref{eq:PerturbationOperator}, and higher-order corrections to the perturbation formula~\eqref{eq:PerturbationLogA}. Although our technique does not give a full nonperturbative result for arbitrary masses, such a result could be obtained if one would show that the perturbative expansion in the mass converges. This would be the case if one could derive suitable bounds on the full perturbation operator~\cite{Friedrichs:1938,Friedrichs:1948}, and we note that such bounds do hold for the kernel~\eqref{eq:PerturbationOperator.MatrixElement.Solution} of the first-order operator, which is a smooth function of fast decay.

Of course, it would be very interesting to extend the perturbative approach to the modular Hamiltonian of fermions on other spacetimes or in other quantum states. For example, one could consider an extension to a spatially periodic, flat spacetime based on the massless modular Hamiltonian reported in~\cite{KlichVamanWong:2017}. Unfortunately, the extension to bosonic fields is plagued with difficulties since the massless limit for a scalar field on Minkowski is infrared-divergent. Therefore, it seems that a different approach would be needed for those cases.

\begin{acknowledgments}
This work has been funded by the Deutsche Forschungsgemeinschaft (DFG, German Research Foundation) --- project no. 396692871 within the Emmy Noether grant CA1850/1-1.
We thank Henning Bostelmann, Horacio Casini and Erik Tonni for discussions about the modular Hamiltonian as well as the presentation and results of this work. We also thank the anonymous referees for their useful comments and suggestions, especially regarding the idea for the proof of Lemma~\ref{lemma:modular}.
\end{acknowledgments}

\appendix

\section{Free fermions}
\label{sec:AppendixFreeFermions}

In this appendix, we include the calculations for the Wightman two-point function for free fermions. We recall that we use the conventions of~\cite{FreedmanVanProeyen:2012}, and we denote spacetime points by $(t,x)$ and $(s,y)$.

The Dirac adjoint, which for the choice~\eqref{eq:DiracMatrices} of the Dirac matrices agrees with the Majorana adjoint, is given by
\begin{equation}
\label{eq:DiracAdjoint}
\bar\psi \coloneqq \mathi \psi^\dagger \gamma^0 = ( - \mathi \psi_2^*, \mathi \psi_1^* ) \eqend{.}
\end{equation}
Writing a dot for a time derivative and a prime for a space derivative, the Dirac action reads
\begin{splitequation}
S &= - \int \bar\psi \left( \gamma^\mu \partial_\mu - m \right) \psi \total t \total x \\
&= \int \left( \mathi \psi_2^* \dot\psi_2 + \mathi \psi_1^* \dot\psi_1 + \mathi \psi_2^* \psi_2' - \mathi \psi_1^* \psi_1' \right) \total^2 x + m \int \left( - \mathi \psi_2^* \psi_1 + \mathi \psi_1^* \psi_2 \right) \total t \total x \eqend{.}
\end{splitequation}
The equation of motion follows as
\begin{equation}
\left( \gamma^\mu \partial_\mu - m \right) \psi = 0 \eqend{,}
\end{equation}
and splits into the two equations
\begin{equation}
\dot\psi_2 + \psi_2' - m \psi_1 = 0 \eqend{,} \quad - \dot\psi_1 + \psi_1' - m \psi_2 = 0 \eqend{.}
\end{equation}

To compute the Wightman two-point function in the Minkowski vacuum, we start from the time-ordered (Feynman) propagator
\begin{equation}
\mathcal{G}^\text{F}(t,x;s,y) = - \mathi \expect{ \mathcal{T} \psi(t,x) \bar\psi(s,y) } \eqend{,}
\end{equation}
which satisfies
\begin{equation}
- \left( \gamma^\mu \partial_\mu - m \right) \mathcal{G}^\text{F}(t,x;s,y) = \delta(t-s) \delta(x-y) \eqend{.}
\end{equation}
In Fourier space, we easily obtain
\begin{splitequation}
\mathcal{G}^\text{F}(t,x;s,y) &= \lim_{\epsilon \to 0^+} \int \frac{\mathi \gamma^\mu p_\mu + m}{p^\mu p_\mu + m^2 - \mathi \epsilon} \mathe^{- \mathi p^0 (t-s) + \mathi p^1 (x-y)} \frac{\total^2 p}{(2\pi)^2} \\
&= \left( \gamma^\mu \partial_\mu + m \right) \lim_{\epsilon \to 0^+} \int \frac{1}{p^\mu p_\mu + m^2 - \mathi \epsilon} \mathe^{- \mathi p^0 (t-s) + \mathi p^1 (x-y)} \frac{\total^2 p}{(2\pi)^2} \eqend{.}
\end{splitequation}
Integrating over $p^0$, we see that the integral is independent of the sign of $t-s$, and we can thus replace $t-s$ by $\abs{t-s}$. We can then close the integration contour in the lower half of the complex plane, picking up the residue of the pole at $p^0 = \sqrt{\omega(p^1)^2 - \mathi \epsilon}$ with $\omega(p^1) = \sqrt{(p^1)^2 + m^2}$. This gives
\begin{equation}
\mathcal{G}^\text{F}(t,x;s,y) = \frac{\mathi}{4 \pi} \left( \gamma^\mu \partial_\mu + m \right) \int \frac{1}{\omega(p)} \mathe^{- \mathi \omega(p) \abs{t-s} + \mathi p (x-y)} \total p \eqend{,}
\end{equation}
where we renamed $p^1 = p$. Since for $t > s$ the time-ordered propagator coincides with the positive frequency two-point function, we have
\begin{equation}
\mathcal{G}^+(t,x;s,y) = \frac{\mathi}{4 \pi} \left( \gamma^\mu \partial_\mu + m \right) \int \frac{1}{\omega(p)} \mathe^{- \mathi \omega(p) (t-s) + \mathi p (x-y)} \total p \eqend{.}
\end{equation}
Inserting the explicit expression of the Dirac adjoint~\eqref{eq:DiracAdjoint}, this yields
\begin{splitequation}
\expect{ \psi(t,x) \psi^\dagger(s,y) } &= - \mathcal{G}^+(t,x;s,y) \gamma^0 \\
&= \frac{1}{4 \pi} \int \frac{1}{\omega(p)} \left( \omega(p) - p \gamma_* - \mathi m \gamma^0 \right) \mathe^{- \mathi \omega(p) (t-s) + \mathi p (x-y)} \total p \eqend{,}
\end{splitequation}
and restricting to the Cauchy surface $t = 0$ we obtain
\begin{equation}
\label{eq:TwoPointFunction.Wightman.FromFourier}
\mathcal{G}(x,y) = \expect{ \psi(0,x) \psi^\dagger(0,y) } = \frac{1}{4 \pi} \int \frac{1}{\omega(p)} \begin{pmatrix} \omega(p) - p & - \mathi m \\ \mathi m & \omega(p) + p \end{pmatrix} \mathe^{\mathi p (x-y)} \total p \eqend{.}
\end{equation}

This integral needs to be interpreted as a distributional Fourier transform, which means that we need to insert a convergence factor $\mathe^{- \epsilon \abs{p}/m}$ and take the limit $\epsilon \to 0^+$ after integration. We change variables to
\begin{equation}
p = m \sinh s \eqend{,} \quad \total p = m \cosh s \total s \eqend{,} \quad \omega(p) = m \cosh s \eqend{,}
\end{equation}
and compute
\begin{splitequation}
\mathcal{G}(x,y) &= \frac{m}{4 \pi} \lim_{\epsilon \to 0^+} \int \begin{pmatrix} \mathe^{-s} & - \mathi \\ \mathi & \mathe^s \end{pmatrix} \mathe^{\mathi \sinh s [ m (x-y) + \mathi \epsilon \sgn(s) ]} \total s \\
&= \frac{m}{2 \pi} \lim_{\epsilon \to 0^+} \begin{pmatrix} I_1\bigl(m (x-y)\bigr) - \mathi I_2\bigl(m (x-y)\bigr) & - \mathi I_3\bigl(m (x-y)\bigr) \\ \mathi I_3\bigl(m (x-y)\bigr) & I_1\bigl(m (x-y)\bigr) + \mathi I_2\bigl(m (x-y)\bigr) \end{pmatrix}
\end{splitequation}
with the integrals
\begin{equations}
I_1(z) &\coloneqq \int_0^\infty \cos\left( z \sinh s \right) \mathe^{- \epsilon \sinh s} \cosh s \total s \eqend{,} \\
I_2(z) &\coloneqq \int_0^\infty \sin\left( z \sinh s \right) \mathe^{- \epsilon \sinh s} \sinh s \total s \eqend{,} \\
I_3(z) &\coloneqq \int_0^\infty \cos\left( z \sinh s \right) \mathe^{- \epsilon \sinh s} \total s \eqend{.}
\end{equations}
For the first integral we change variables to $\sinh s = t$ and compute
\begin{splitequation}
I_1(z) = \int_0^\infty \cos(z t) \, \mathe^{- \epsilon t} \total t = \frac{\epsilon}{z^2 + \epsilon^2} \to \pi \delta(z) \quad (\epsilon \to 0^+) \eqend{,}
\end{splitequation}
and we note that
\begin{equation}
I_2(z) = - \partial_z I_3(z) \eqend{.}
\end{equation}
In $I_3$ (in the limit $\epsilon \to 0^+$) we recognise the integral representation of the modified Bessel function~\citeDLMFeq{10.32}{6}, such that
\begin{equation}
\lim_{\epsilon \to 0^+} I_3(z) = \bessel{K}0{\abs{z}} \eqend{.}
\end{equation}
For small $z$, we have~\citeDLMFeq{10.31}{2}
\begin{equation}
\bessel{K}0{z} \sim - \gamma_\mathrm{E} - \ln z + \ln 2 + \bigo{z^2 \ln z} \eqend{,}
\end{equation}
and thus
\begin{splitequation}
I_2(z) &= - \partial_z \bessel{K}0{\abs{z}} = - \partial_z \Bigl( \bessel{K}0{\abs{z}} + \ln\abs{z} \Bigr) + \partial_z \ln\abs{z} \\
&= \left( \sgn(z) \bessel{K}1{\abs{z}} - \frac{1}{z} \right) + \pvalue \frac{1}{z} \eqend{,}
\end{splitequation}
where $\pvalue$ denotes the Cauchy principal value. With the results for the integrals $I_i$ and the Sochocki--Plemelj formula
\begin{equation}
\label{eq:SochockiPlemelj}
\lim_{\epsilon \to 0^+} \frac{1}{x-y \pm \mathi \epsilon} = \pvalue \frac{1}{x-y} \mp \mathi \pi \delta(x-y) \eqend{,}
\end{equation}
we then obtain the components~\eqref{eq:TwoPointFunction.Components} of the two-point function.

The two-point function $\mathcal{G}$ also determines the inner product on the complexified one-particle Hilbert space $\mathcal{H} = L^2_\mathbb{R}(\mathbb{R}) \oplus L^2_\mathbb{R}(\mathbb{R})$ of initial data at $t = 0$, which is given by
\begin{equation}
\expect{ f, g }_\mathcal{H} \coloneqq 2 \sum_{a,b} \iint f_a(x) \mathcal{G}_{ab}(x,y) g_b(y) \total x \total y \eqend{.}
\end{equation}
The standard $L^2$ scalar product
\begin{equation}
\left( f, g \right) \coloneqq \sum_a \int f_a(x) g_a(x) \total x
\end{equation}
is the real part of $\expect{ \cdot, \cdot }_\mathcal{H}$, and the imaginary part of $\expect{ \cdot, \cdot }_\mathcal{H}$ is given by
\begin{equation}
- \sum_{a,b} \iint f_a(x) I_{ab}(x,y) g_b(y) \total x \total y = - \left( f, I g \right)
\end{equation}
with the kernel $I_{ab}(x,y)$ of the complex structure $I$, which acts as a convolution operator on initial data. We thus have the relation~\eqref{eq:ComplexStructure} $I = - \mathi ( \1 - 2 \mathcal{G} )$, and using the explicit form~\eqref{eq:TwoPointFunction.Wightman.FromFourier} of the two-point function we obtain
\begin{equation}
\label{eq:ComplexStructure.FourierSpace}
I(x,y) = \int \frac{1}{\omega(p)} \begin{pmatrix} - \mathi p & m \\ - m & \mathi p \end{pmatrix} \mathe^{\mathi p (x-y)} \frac{\total p}{2 \pi} = \frac{1}{\pi} \begin{pmatrix} - \partial_x & m \\ - m & \partial_x \end{pmatrix} \bessel{K}0{m \abs{x-y}} \eqend{.}
\end{equation}
This is clearly real-valued and anti-hermitean (with respect to the real $L^2$ scalar product), and we compute
\begin{equation}
I^2(x,y) = \int \left[ \frac{1}{\omega(p)} \begin{pmatrix} - \mathi p & m \\ - m & \mathi p \end{pmatrix} \right]^2 \mathe^{\mathi p (x-y)} \frac{\total p}{2 \pi} = - \delta(x-y) \begin{pmatrix} 1 & 0 \\ 0 & 1 \end{pmatrix} \eqend{,}
\end{equation}
such that $I$ indeed satisfies the properties of a complex structure. We thus have
\begin{equation}
\expect{ f, g }_\mathcal{H} = \left( f, g \right) - \mathi \left( f, I g \right) \eqend{,}
\end{equation}
and using this decomposition and the above properties of $I$, it is straightforward to verify that
\begin{equation}
\expect{ f, I g }_\mathcal{H} = \mathi \expect{ f, g }_\mathcal{H} \eqend{,} \quad \expect{ I f, g }_\mathcal{H} = - \mathi \expect{ f, g }_\mathcal{H} \eqend{,}
\end{equation}
such that the scalar product $\expect{ \cdot, \cdot }_\mathcal{H}$ is compatible with the complex structure. Furthermore, since $I$ in Fourier space~\eqref{eq:ComplexStructure.FourierSpace} is given by a polynomial in $p$ times $1/\omega(p) = \left( p^2 + m^2 \right)^{- \frac{1}{2}}$, it follows from~\cite[Remark]{SegalGoodman:1965} (or by adapting the arguments in~\cite{Masuda:1972,Murata:1973}) that $I$ is an antilocal operator. That is, any twice continuously differentiable function $f \in C^2_\mathbb{R}(\mathbb{R}) \oplus C^2_\mathbb{R}(\mathbb{R})$ for which both $f$ and $I f$ vanish in some interval is identically zero. Since it is well known that $C^2_\mathbb{R}(\mathbb{R}) \oplus C^2_\mathbb{R}(\mathbb{R})$ is dense in $\mathcal{H}$, we see that the set of vectors $h \in \mathcal{H}$ for which $(\1-E) I E h \neq 0$ is dense in the subspace $\mathcal{L} = E \mathcal{H}$, where we recall that $E$ is the orthogonal projection on the interval $V = [-\ell,\ell]$.

\section{Computations for the massless decomposition}
\label{appx:MasslessDecomposition}

In this appendix, we include all calculations for the results on the massless modular Hamiltonian of a double cone in $(1+1)$-dimensional Minkowski spacetime as given in section~\ref{sec:ModularHamiltonian.Massless}.

We recall that the components of the massless two-point function $\mathcal{G}_V$ with integral kernels given by~\eqref{eq:TwoPointFunction.Components.Massless} have the generalized eigenvectors $\varPsi^{(k)}_a(s,x) = \delta^k_a \varPsi(s,x)$~\eqref{eq:MasslessEigenvector} for $x \in (-\ell,\ell)$. The functions
\begin{equation}
\varPsi(s,x) = \sqrt{ \frac{\ell}{\pi} } \frac{1}{\sqrt{\ell^2 - x^2}} \left( \frac{\ell - x}{\ell + x} \right)^{\mathi s}
\end{equation}
satisfy properties~\eqref{eq:MasslessEigenvector.Properties}, $\varPsi^*(s,x) = \varPsi(-s,x) = \varPsi(s,-x)$.
The corresponding eigenvalues read~\eqref{eq:MasslessEigenvalue.Properties}
\begin{equation}
\lambda^{(1)}(s) = \frac{1}{1 + \mathe^{- 2 \pi s}} \eqend{,} \quad \lambda^{(2)}(s) = \frac{1}{1 + \mathe^{2 \pi s}} = \lambda^{(1)}(-s) \eqend{.}
\end{equation}

To verify that these are the correct eigenvectors and eigenvalues, let 
\begin{equation}
\label{eq:MasslessEigenvector.Integral}
\varUpsilon_\pm( s, x ) \coloneqq \mp \lim_{\epsilon \to 0^+} \int_{-\ell}^\ell \frac{1}{x - y \mp \mathi \epsilon} \varPsi(s,y) \total y \eqend{.}
\end{equation}
We only need to show that 
\begin{equation}
\label{eq:MasslessEigenvector.IntegralMinus.Solution}
\varUpsilon_-( s, x ) = - \frac{2 \pi \mathi}{1 + \mathe^{2 \pi s}} \varPsi(s,x) \eqend{,}
\end{equation}
since the second eigenvector integral $\varUpsilon_+( s, x )$ follows directly from this result, using that the two integral kernels are related by the Sochocki--Plemelj formula~\eqref{eq:SochockiPlemelj} according to $\varUpsilon_+( s, x ) = - \varUpsilon_-( s, x ) - 2 \pi \mathi \varPsi( s, x )$. 

Hence, it suffices to prove \eqref{eq:MasslessEigenvector.IntegralMinus.Solution}, for which we perform the variable transformation
\begin{equation}
\label{eq:VariableTransformation}
y = \ell \tanh(\pi w) \eqend{,} \quad \total y = \pi \ell \cosh^{-2}(\pi w) \total w \eqend{.}
\end{equation}
This results in
\begin{splitequation}
\varUpsilon_-( s, x ) &= \sqrt{ \ell \pi } \lim_{\epsilon \to 0^+} \int_{-\infty}^\infty \frac{1}{x - \ell \tanh(\pi w) + \mathi \epsilon} \frac{\mathe^{- 2 \pi \mathi s w}}{\cosh(\pi w)} \total w \\
&= \sqrt{ \ell \pi } \lim_{\epsilon \to 0^+} \int_{-\infty}^\infty \xi_\epsilon(s,w,x) \total w \eqend{,}
\end{splitequation}
where we defined
\begin{equation}
\xi_\epsilon(s,w,x) \coloneqq \frac{\mathe^{- 2 \pi \mathi s w}}{x \cosh(\pi w) - \ell \sinh(\pi w) + \mathi \epsilon} \eqend{.}
\end{equation}
To evaluate the integral, we use the Cauchy residue theorem and integrate over the contour depicted in figure~\ref{fig:MasslessContour-w-plane}. For $s < 0$, we close the contour in the upper half plane (red), while for $s > 0$ we close the contour in the lower half plane (blue).
\begin{figure}
    \includegraphics{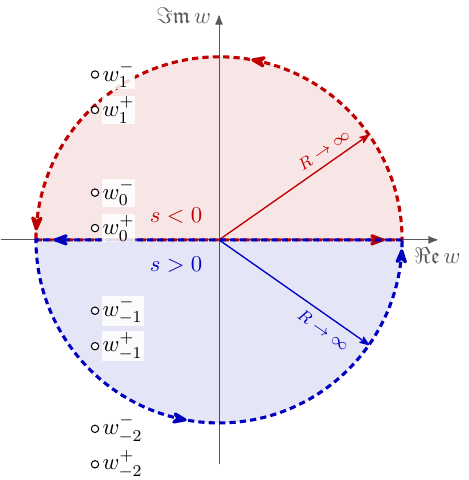}
    \caption{The integration contour (dashed) in the upper half plane (red) has a vanishing contribution for $s < 0$ and the one in the lower half plane (blue) has a vanishing contribution for $s > 0$ along the respective semicircles. Either contour integral passes over a sequence of poles $w^\pm_k$ for $k \geq 0$ or $k < 0$, respectively, since $\epsilon > 0$. In the limit $\epsilon \to 0^+$, the real part of the poles reads $\frac{1}{2 \pi} \ln\left( \frac{\ell + x}{\ell - x} \right)$.}
    \label{fig:MasslessContour-w-plane}
\end{figure}
The integrand $\xi_\epsilon(s,w,x)$ has an infinite series of poles at the points
\begin{splitequation}
w_k^\pm &= \frac{1}{\pi} \ln\left( \frac{\sqrt{\ell^2-x^2-\epsilon^2} \pm \mathi \epsilon}{\ell-x} \right) + 2 \mathi k + \frac{\mathi}{2} (1 \mp 1) \eqend{,} \quad k \in \mathbb{Z} \eqend{,}
\end{splitequation}
of which the $w_k^+$ with $k \geq 0$ and the $w_k^-$ with $k \geq 0$ lie in the upper half plane, and the $w_k^+$ with $k < 0$ and the $w_k^-$ with $k < 0$ lie in the lower half plane. We thus obtain
\begin{splitequation}
\varUpsilon_-( s, x ) &= 2 \pi \mathi \sqrt{ \ell \pi } \lim_{\epsilon \to 0^+} \Biggl[ \Theta(-s) \left( \sum_{k=0}^\infty \Res_{w = w_k^+} + \sum_{k=0}^\infty \Res_{w = w_k^-} \right) \xi_\epsilon(s,w,x) \\
&\hspace{7em}- \Theta(s) \left( \sum_{k=1}^\infty \Res_{w = w_{-k}^+} + \sum_{k=1}^\infty \Res_{w = w_{-k}^-} \right) \xi_\epsilon(s,w,x) \Biggr] \\
&= 2 \mathi \sqrt{ \ell \pi } \lim_{\epsilon \to 0^+} \Biggl[ \Theta(-s) \sum_{k=0}^\infty \frac{\mathe^{- 2 \pi \mathi s w_k^-} - \mathe^{- 2 \pi \mathi s w_k^+}}{\sqrt{\ell^2 - x^2 - \epsilon^2}} \\
&\hspace{7em}+ \Theta(s) \sum_{k=1}^\infty \frac{\mathe^{- 2 \pi \mathi s w_{-k}^+} - \mathe^{- 2 \pi \mathi s w_{-k}^-}}{\sqrt{\ell^2 - x^2 - \epsilon^2}} \Biggr] \\
&= \frac{2 \mathi \sqrt{ \ell \pi }}{\sqrt{\ell^2 - x^2}}
\Bigg[ \Theta(-s) \sum_{k=0}^\infty \left( \mathe^{- \mathi s \ln\left( \frac{\ell+x}{\ell-x} \right) + 4 \pi s k + 2 \pi s} - \mathe^{- \mathi s \ln\left( \frac{\ell+x}{\ell-x} \right) + 4 \pi s k} \right) \\
&\hspace{7em}+ \Theta(s) \sum_{k=1}^\infty \left( \mathe^{- \mathi s \ln\left( \frac{\ell+x}{\ell-x} \right) - 4 \pi s k} - \mathe^{- \mathi s \ln\left( \frac{\ell+x}{\ell-x} \right) - 4 \pi s k + 2 \pi s} \right) \Bigg] \\
&= - \frac{2 \mathi \sqrt{ \ell \pi }}{\sqrt{\ell^2 - x^2}} \frac{\mathe^{- \mathi s \ln\left( \frac{\ell+x}{\ell-x} \right)}}{1 + \mathe^{2 \pi s}} \\
&= - \frac{2 \pi \mathi}{1 + \mathe^{2 \pi s}} \varPsi(s,x) \eqend{,}
\end{splitequation}
which is exactly what we wanted to prove.

We also need to show that the eigenvector functions $\varPsi( s, x )$ are orthogonal, normalised in the distributional sense,
\begin{equation}
\int_{-\ell}^\ell \varPsi^*(s,x) \varPsi(t,x) \total x = \delta(s-t) \eqend{,}
\end{equation}
and that they form a complete basis,
\begin{equation}
\int_{-\infty}^\infty \varPsi^*(s,x) \varPsi(s,y) \total s = \delta(x-y) \eqend{,}
\end{equation}
since these properties then imply that the generalised eigenvectors are orthonormal and complete according to~\eqref{eq:Eigenvector.Properties}.

To show orthogonality, we perform the same change of variable~\eqref{eq:VariableTransformation} (with $x$ instead of $y$) and obtain the required result
\begin{splitequation}
\int_{-\ell}^\ell \varPsi^*(s,x) \varPsi(t,x) \total x
&= \frac{\ell}{\pi} \int_{-\ell}^\ell \frac{\left( \frac{\ell + x}{\ell - x} \right)^{\mathi s}}{\sqrt{\ell^2 - x^2}} \frac{\left( \frac{\ell - x}{\ell + x} \right)^{\mathi t}}{\sqrt{\ell^2 - x^2}} \total x \\
&= \int_{-\infty}^\infty \mathe^{2 \pi \mathi w (s-t)} \total w \\
&= \delta(s-t) \eqend{.}
\end{splitequation}
To show completeness, we compute
\begin{splitequation}
\int_{-\infty}^\infty \varPsi^*(s,x) \varPsi(s,y) \total s
&= \frac{\ell}{\pi} \int_{-\infty}^\infty \frac{\left( \frac{\ell + x}{\ell - x} \right)^{\mathi s}}{\sqrt{\ell^2 - x^2}} \frac{\left( \frac{\ell - y}{\ell + y} \right)^{\mathi s}}{\sqrt{\ell^2 - y^2}} \total s \\
&= \frac{\ell}{\pi} \frac{1}{\sqrt{ (\ell^2-x^2) (\ell^2-y^2) }} \int_{-\infty}^\infty \exp\left[ \mathi s \ln\left( \frac{\ell+x}{\ell-x} \frac{\ell-y}{\ell+y} \right) \right] \total s \\
&= 2 \ell \frac{\delta\left[ \ln\left( \frac{\ell+x}{\ell-x} \frac{\ell-y}{\ell+y} \right) \right]}{\sqrt{ (\ell^2-x^2) (\ell^2-y^2) }} \\
&= \delta(x-y) \eqend{.}
\end{splitequation}
Note that the integrals converge only in a distributional sense, which was to be expected since we are considering generalised eigenfunctions. In the last step, we also used the composition formula
\begin{equation}
\label{eq:Delta.Composition}
\delta\bigl[ f(x) \bigr] = \sum_i \frac{1}{\abs{ f'(x_i) }} \delta(x-x_i) \eqend{,}
\end{equation}
which holds for any smooth function $f$ with simple zeros $x_i$. For any fixed value $y \in (-\ell,\ell)$, the logarithm $f(x) = \ln\left( \frac{\ell + x}{\ell - x} \frac{\ell - y}{\ell + y} \right)$ is defined for all values $x \in (-\ell,\ell)$ and it has exactly one simple zero at $x = y$, such that this formula was applicable.

Next, we compute the integral for the propagator~\eqref{eq:TwoPointFunction.Component11.SpectralDecomposition} that is of the form 
\begin{equation}
\int_{-\infty}^\infty \frac{1}{c + \mathe^{2 \pi s}} \mathe^{\mathi s z} \total s
\end{equation}
with a constant $c > 0$. 
After inserting a convergence factor $\mathe^{\epsilon s}$, we can again use the Cauchy theorem and close the contour in either the upper or lower half plane depending on the sign of $z$. 
There is an infinite series of poles at
\begin{equation}
s_k = \frac{\log c}{2 \pi} + \mathi \left( k + \frac{1}{2} \right) \eqend{,}
\end{equation}
as shown in figure~\ref{fig:MasslessContour-s-plane}. 
\begin{figure}
    \includegraphics{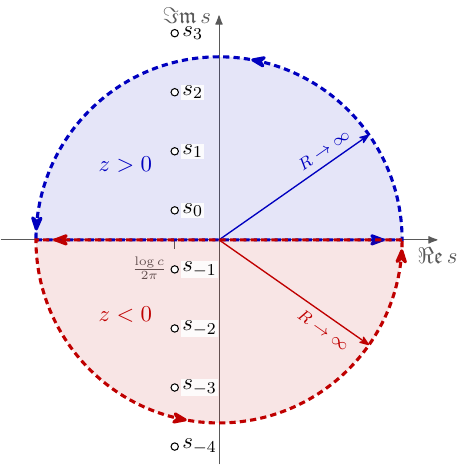}
    \caption{The integration contour (dashed) in the upper half plane (blue) has a vanishing contribution for $z > 0$ and the one in the lower half plane (red) has a vanishing contribution for $z < 0$ along the respective semicircles. Either contour integral passes over a sequence of poles $s_k$ for $k \geq 0$ or $k < 0$, respectively, when $\epsilon > 0$.}
    \label{fig:MasslessContour-s-plane}
\end{figure}
Summing the residues yields 
\begin{splitequation}
&\int_{-\infty}^\infty \frac{1}{c + \mathe^{2 \pi s}} \mathe^{\mathi s z} \total s \\
&\quad= 2 \pi \mathi \lim_{\epsilon \to 0^+} \bigg[ \Theta(z) \sum_{k=0}^\infty \Res_{s = s_k} \frac{\mathe^{\epsilon s}}{c + \mathe^{2 \pi s}} \mathe^{\mathi s z} - \Theta(-z) \sum_{k=1}^\infty \Res_{s = s_{-k}} \frac{\mathe^{\epsilon s}}{c + \mathe^{2 \pi s}} \mathe^{\mathi s z} \bigg] \\
&\quad= \frac{\mathi}{c} \lim_{\epsilon \to 0^+} \bigg[ - \Theta(z) c^\frac{\epsilon + \mathi z}{2 \pi} \sum_{k=0}^\infty \mathe^{- (z - \mathi \epsilon) \left( k + \frac{1}{2} \right)} + \Theta(-z) c^\frac{\epsilon + \mathi z}{2 \pi} \sum_{k=1}^\infty \mathe^{- (z - \mathi \epsilon) \left( - k + \frac{1}{2} \right)} \bigg] \\
&\quad= - \frac{\mathi}{c} \lim_{\epsilon \to 0^+} \left[ c^\frac{\epsilon + \mathi z}{2 \pi} \frac{\mathe^\frac{z - \mathi \epsilon}{2}}{\mathe^{z - \mathi \epsilon} - 1} \right] \\
&\quad= - \frac{\mathi}{2 c} \lim_{\epsilon \to 0^+} \left[ c^\frac{\epsilon + \mathi z}{2 \pi} \sinh^{-1}\left( \frac{z - \mathi \epsilon}{2} \right) \right] \eqend{.}
\end{splitequation}

Lastly, we compute the integral in expression~\eqref{eq:ModularHamiltonian.Interval.Massless} for the massless modular Hamiltonian, which is
\begin{splitequation}
H_{ab}(x,y) &= 2 \ell \gamma_{*ab} \int_{-\infty}^\infty \frac{s}{\sqrt{ (\ell^2-x^2) (\ell^2-y^2) }} \left( \frac{\ell - x}{\ell + x} \right)^{\mathi s} \left( \frac{\ell + y}{\ell - y} \right)^{\mathi s} \total s \\
&\quad= \frac{2 \ell \gamma_{*ab}}{\sqrt{ (\ell^2-x^2) (\ell^2-y^2) }} \left[ \mathi \partial_z \int_{-\infty}^\infty \mathe^{- \mathi s z} \total s \right]_{z = \ln\left( \frac{\ell+x}{\ell-x} \frac{\ell-y}{\ell+y} \right)} \\
&\quad= \frac{4 \mathi \ell \pi \gamma_{*ab}}{\sqrt{ (\ell^2-x^2) (\ell^2-y^2) }} \delta'\left[ \ln\left( \frac{\ell+x}{\ell-x} \frac{\ell-y}{\ell+y} \right) \right] \\
&\quad= 4 \mathi \ell \pi \gamma_{*ab} \sqrt{ \frac{\ell^2-y^2}{\ell^2-x^2} } \frac{2 y \delta(x-y) + (\ell^2-y^2) \delta'(x-y)}{4 \ell^2} \\
&\quad= 2 \ell \gamma_{*ab} \pi \mathi \frac{2 \ell^2 - x^2 - y^2}{4 \ell^2} \delta'(x-y)
\end{splitequation}
using the generalisation of the above composition formula
\begin{equation}
\label{eq:DeltaDeriv.Composition}
\delta'(f(x)) = \sum_i \frac{f''(x_i) \delta(x-x_i) + f'(x_i) \delta'(x-x_i)}{\abs{f'(x_i)}^3} \eqend{,}
\end{equation}
valid for any smooth function $f$ with simple zeros $x_i$. Again this condition is fulfilled, the integral converged in a distributional sense, and we have used in addition that
\begin{equation}
\label{eq:DeltaDeriv.Product}
f(y) \delta'(x-y) = f(x) \delta'(x-y) + f'(x) \delta(x-y) \eqend{,}
\end{equation}
as can be easily verified by integrating with a test function in $y$.

\section{Finite-part distributions}
\label{appx:FinitePartDistribution}

In this section, we give details on the distribution~\eqref{eq:ModifiedPrincipalValueDistribution}
\begin{equation}
\pvalue_\mu \frac{1}{\abs{x}} \coloneqq \lim_{\epsilon \to 0^+} \left[ \frac{\Theta(x-\epsilon) - \Theta(-x-\epsilon)}{x} + 2 \ln\left( \mu \epsilon \, \mathe^{-\gamma_\mathrm{E}} \right) \delta(x) \right] \eqend{,}
\end{equation}
which depends on a parameter $\mu$ with dimensions of mass to make the logarithm dimensionless, and related distributions. We first show that $\pvalue_\mu \frac{1}{\abs{x}}$ is a well-defined distribution, by letting it act on a test function $f$. We integrate by parts, change $x \to -x$ in the second term, and simplify to obtain
\begin{splitequation}
\int f(x) \pvalue_\mu \frac{1}{\abs{x}} \total x &= \lim_{\epsilon \to 0^+} \int f(x) \left[ \frac{\Theta(x-\epsilon) - \Theta(-x-\epsilon)}{x} + 2 \ln\left( \mu \epsilon \, \mathe^{-\gamma_\mathrm{E}} \right) \delta(x) \right] \total x \\
&= \lim_{\epsilon \to 0^+} \left[ \int_\epsilon^\infty \frac{f(x)}{x} \total x - \int_{-\infty}^{-\epsilon} \frac{f(x)}{x} \total x + 2 \ln\left( \mu \epsilon \, \mathe^{-\gamma_\mathrm{E}} \right) f(0) \right] \\
&= \lim_{\epsilon \to 0^+} \Big\lbrack - \int_\epsilon^\infty [ f'(x) - f'(-x) ] \ln(\mu x) \total x - f(\epsilon) \ln( \mu \epsilon ) \\
&\quad \quad \quad \quad - f(-\epsilon) \ln( \mu \epsilon ) + 2 \ln\left( \mu \epsilon \, \mathe^{-\gamma_\mathrm{E}} \right) f(0) \Big\rbrack \eqend{.}
\end{splitequation}
Since the logarithm is an integrable function, we take the limit $\epsilon \to 0^+$, and then all terms containing $\ln \epsilon$ cancel such that
\begin{equation}
\int f(x) \pvalue_\mu \frac{1}{\abs{x}} \total x = - \int_0^\infty [ f'(x) - f'(-x) ] \ln(\mu x) \total x - 2 \gamma_\mathrm{E} f(0) \eqend{,}
\end{equation}
showing that $\pvalue_\mu \frac{1}{\abs{x}}$ is a well-defined distribution. 

Analogously, we define the distributions
\begin{equation}
\pvalue_\mu \frac{\Theta(\pm x)}{x} \coloneqq \lim_{\epsilon \to 0^+} \left[ \frac{\Theta(\pm x-\epsilon)}{x} \pm \ln\left( \mu \epsilon \, \mathe^{-\gamma_\mathrm{E}} \right) \delta(x) \right] \eqend{,}
\end{equation}
and compute that
\begin{equation}
\label{eq:PfTheta.Appendix}
\int f(x) \pvalue_\mu \frac{\Theta(\pm x)}{x} \total x = - \int_0^\infty f'(\pm x) \ln(\mu x) \total x \mp \gamma_\mathrm{E} f(0) \eqend{.}
\end{equation}
By construction, we have
\begin{equation}
\label{eq:PfTheta.Relation.Appendix}
\pvalue_\mu \frac{1}{\abs{x}} = \pvalue_\mu \frac{\Theta(x)}{x} - \pvalue_\mu \frac{\Theta(-x)}{x} \eqend{,}
\end{equation}
and all of these distributions are almost homogeneous of degree $-1$, which can be computed directly from the definition. Namely, for $\lambda > 0$ we obtain
\begin{splitequation}
\label{eq:PfTheta.Scaling.Appendix}
\pvalue_\mu \frac{\Theta(\pm \lambda x)}{\lambda x} &= \lim_{\epsilon \to 0^+} \left[ \frac{\Theta(\pm \lambda x-\epsilon)}{\lambda x} \pm \ln\left( \mu \epsilon \, \mathe^{-\gamma_\mathrm{E}} \right) \delta(\lambda x) \right] \\
&= \lambda^{-1} \lim_{\lambda \epsilon \to 0^+} \left[ \frac{\Theta(\pm \lambda x - \lambda \epsilon)}{x} \pm \ln\left( \mu \lambda \epsilon \, \mathe^{-\gamma_\mathrm{E}} \right) \delta(x) \right] \\
&= \lambda^{-1} \lim_{\epsilon \to 0^+} \left[ \frac{\Theta(\pm x - \epsilon)}{x} \pm \ln\left( \mu \lambda \epsilon \, \mathe^{-\gamma_\mathrm{E}} \right) \delta(x) \right] \\
&= \lambda^{-1} \left[ \pvalue_\mu \frac{\Theta(\pm x)}{x} \pm \ln \lambda \, \delta(x) \right] \eqend{,}
\end{splitequation}
where we relabeled $\epsilon \to \lambda \epsilon$ in passing from the first to the second equality, and used formula~\eqref{eq:Delta.Composition}. Since $\lambda > 0$ is fixed, the limit $\lambda \epsilon \to 0^+$ is the same as $\epsilon \to 0^+$, and we could drop $\lambda$ from the limit in the next equality. The almost homogeneous scaling is then expressed by the relation
\begin{equation}
\left( \lambda \partial_\lambda \right)^2 \left[ \lambda \pvalue_\mu \frac{\Theta(\pm \lambda x)}{\lambda x} \right] = 0 \eqend{;}
\end{equation}
for a homogeneous scaling the power of the Euler operator $\lambda \partial_\lambda$ would have to be equal to $1$. Analogously, one obtains
\begin{equation}
\pvalue_\mu \frac{1}{\abs{\lambda x}} = \lambda^{-1} \left[ \pvalue_\mu \frac{1}{\abs{x}} + 2 \ln \lambda \, \delta(x) \right] \eqend{.}
\end{equation}
We are also interested in changes of the scale parameter $\mu$. These are even easier to compute, and we obtain
\begin{equations}[eq:PfTheta.ParameterScaling.Appendix]
\pvalue_{\lambda\mu} \frac{\Theta(\pm x)}{x} &= \pvalue_\mu \frac{\Theta(\pm x)}{x} \pm \ln \lambda \, \delta(x) \eqend{,} \\
\pvalue_{\lambda\mu} \frac{1}{\abs{x}} &= \pvalue_\mu \frac{1}{\abs{x}} + 2 \ln \lambda \, \delta(x) \eqend{.}
\end{equations}

Later on, we also need to compute distributional limits where the Heaviside $\Theta$ distribution has a more complicated dependence on $x$, namely the limit
\begin{equation}
\lim_{\epsilon \to 0^+} \frac{1}{x} \Theta\left( x - \left[ (1+x)^2 - a^2 \right] \epsilon \right)
\end{equation}
for $a^2 < 1$. Smearing with a test function $f \in \mathcal{S}(\mathbb{R})$, we compute
\begin{splitequation}
\label{eq:PfThetaCoordinate.Appendix}
&\int f(x) \frac{1}{x} \Theta\left( x - \left[ (1+x)^2 - a^2 \right] \epsilon \right) \total x \\
&\quad= \int f(x) \partial_x \ln(\mu x) \Theta\left( x - \left[ (1+x)^2 - a^2 \right] \epsilon \right) \total x \\
&\quad= - \int f'(x) \ln(\mu x) \Theta\left( x - \left[ (1+x)^2 - a^2 \right] \epsilon \right) \total x \\
&\qquad- \int f(x) \ln(\mu x) \left[ 1 - 2 (1+x) \epsilon \right] \delta\left( x - \left[ (1+x)^2 - a^2 \right] \epsilon \right) \total x \eqend{,}
\end{splitequation}
where we did an integration by parts without boundary terms since the test function $f$ decays faster than any polynomial at infinity, and where $\mu$ is a parameter. The last integral can be evaluated using the composition formula~\eqref{eq:Delta.Composition} for the Dirac $\delta$ distribution, and results in
\begin{splitequation}
&\int f(x) \ln(\mu x) \left[ 1 - 2 (1+x) \epsilon \right] \delta\left( x - \left[ (1+x)^2 - a^2 \right] \epsilon \right) \total x \\
&\quad= - f(x_+) \ln(\mu x_+) + f(x_-) \ln(\mu x_-) \eqend{,}
\end{splitequation}
where we defined
\begin{equation}
x_\pm \coloneqq \frac{1 - 2 \epsilon \pm \sqrt{ 1 - 4 \epsilon + 4 \epsilon^2 a^2 }}{2 \epsilon} \eqend{.}
\end{equation}
For small $\epsilon$, we have $x_- = (1-a^2) \epsilon + \bigo{\epsilon^2}$ and $x_+ = \frac{1}{\epsilon} + \bigo{\epsilon^0}$. Since the test function $f$ decays faster than any polynomial at infinity, in the limit $\epsilon \to 0^+$ the term with $f(x_+)$ vanishes. On the other hand, the term with $f(x_-)$ has a logarithmic divergence as $\epsilon \to 0^+$. 
Since in the first integral in the last equality of~\eqref{eq:PfThetaCoordinate.Appendix} the logarithm $\ln(\mu x)$ is integrable at $x = 0$, we can take the limit $\epsilon \to 0^+$ inside that integral. Taking all together, it follows that
\begin{splitequation}
\label{eq:PfThetaCoordinate2.Appendix}
&\lim_{\epsilon \to 0^+} \left[ \int f(x) \frac{1}{x} \Theta\left( x - \left[ (1+x)^2 - a^2 \right] \epsilon \right) \total x + \ln\left[ \mu (1-a^2) \epsilon \right] f(0) \right] \\
&\quad= - \int_0^\infty f'(x) \ln(\mu x) \total x = \int f(x) \pvalue_\mu \frac{\Theta(x)}{x} \total x + \gamma_\mathrm{E} f(0) \eqend{,}
\end{splitequation}
where we compared with the result~\eqref{eq:PfTheta.Appendix}. We have thus proven the distributional limit
\begin{equation}
\label{eq:PfThetaCoordinate.Limit.Appendix}
\lim_{\epsilon \to 0^+} \left[ \frac{1}{x} \Theta\left( x - \left[ (1+x)^2 - a^2 \right] \epsilon \right) + \ln\left[ \mu (1-a^2) \epsilon \, \mathe^{-\gamma_\mathrm{E}} \right] \delta(x) \right] = \pvalue_\mu \frac{\Theta(x)}{x} \eqend{,}
\end{equation}
and an analogous computation establishes that
\begin{equation}
\label{eq:PfThetaCoordinate.Limit2.Appendix}
\lim_{\epsilon \to 0^+} \left[ \frac{1}{x} \Theta\left( - x - \left[ (1-x)^2 - a^2 \right] \epsilon \right) - \ln\left[ \mu (1-a^2) \epsilon \, \mathe^{-\gamma_\mathrm{E}} \right] \delta(x) \right] = \pvalue_\mu \frac{\Theta(-x)}{x} \eqend{.}
\end{equation}

\section{Computations for the first-order massive corrections}
\label{appx:FirstOrderMassiveCorrection}

In this appendix, we include all details of the computations for the main results presented in section~\ref{sec:MassiveModularHamiltonian}. 

First, we compute the coefficient of the perturbation operator~\eqref{eq:PerturbationOperator.MatrixElement} from subsection~\ref{sec:PerturbationOperator}. As in the massless case, we transform the integration variables 
\begin{splitequation}
x &= \ell \tanh(\pi v) \eqend{,} \quad \total x = \ell \pi \cosh^{-2}(\pi v) \total v \eqend{,} \\
y &= \ell \tanh(\pi w) \eqend{,} \quad \total y = \ell \pi \cosh^{-2}(\pi w) \total w \eqend{,}
\end{splitequation}
and obtain
\begin{splitequation}
\label{eq:PerturbationOperator.MatrixElement.Integrating}
K(s,t) &= \frac{\mathi}{2} \iint_{-\infty}^\infty \ln\left( \frac{m \ell \abs{\tanh(\pi v) - \tanh(\pi w)}}{2} \mathe^{\gamma_\mathrm{E}} \right) \frac{\mathe^{2 \pi \mathi v s - 2 \pi \mathi w t}}{\cosh(\pi v) \cosh(\pi w)} \total v \total w \\
&= \frac{\mathi}{2} \ln\left( m \ell \, \mathe^{\gamma_\mathrm{E}} \right) \iint_{-\infty}^\infty \frac{\mathe^{2 \pi \mathi v s - 2 \pi \mathi w t}}{\cosh(\pi v) \cosh(\pi w)} \total v \total w \\
&\quad- \frac{\mathi}{2} \iint_{-\infty}^\infty \ln\left( 1 + \mathe^{2 \pi v} \right) \frac{\mathe^{2 \pi \mathi v s - 2 \pi \mathi w t}}{\cosh(\pi v) \cosh(\pi w)} \total v \total w \\
&\quad- \frac{\mathi}{2} \iint_{-\infty}^\infty \ln\left( 1 + \mathe^{2 \pi w} \right) \frac{\mathe^{2 \pi \mathi v s - 2 \pi \mathi w t}}{\cosh(\pi v) \cosh(\pi w)} \total v \total w \\
&\quad+ \frac{\mathi}{4} \iint_{-\infty}^\infty \ln\left[ \left( \mathe^{2 \pi v} - \mathe^{2 \pi w} \right)^2 \right] \frac{\mathe^{2 \pi \mathi v s - 2 \pi \mathi w t}}{\cosh(\pi v) \cosh(\pi w)} \total v \total w \\
&= \frac{\mathi}{2} \ln\left( m \ell \, \mathe^{\gamma_\mathrm{E}} \right) \frac{1}{\cosh(\pi s) \cosh(\pi t)} \\
&\quad- \frac{\mathi}{2} \frac{1}{\cosh(\pi t)} \int_{-\infty}^\infty \ln\left( 1 + \mathe^{2 \pi v} \right) \frac{\mathe^{2 \pi \mathi v s}}{\cosh(\pi v)} \total v \\
&\quad- \frac{\mathi}{2} \frac{1}{\cosh(\pi s)} \int_{-\infty}^\infty \ln\left( 1 + \mathe^{2 \pi w} \right) \frac{\mathe^{- 2 \pi \mathi w t}}{\cosh(\pi w)} \total w \\
&\quad+ \frac{\mathi}{4} \iint_{-\infty}^\infty \ln\left[ \left( \mathe^{2 \pi v} - \mathe^{2 \pi w} \right)^2 \right] \frac{\mathe^{2 \pi \mathi v s - 2 \pi \mathi w t}}{\cosh(\pi v) \cosh(\pi w)} \total v \total w \eqend{.}
\end{splitequation}

To solve the remaining integrals, we again use contour integration. For the two single integrals, we use the Mellin--Barnes representation of the logarithm (combining~\citeDLMFeq{15.6}{6} with~\citeDLMFeq{15.4}{1})
\begin{equation}
\ln(1+z) = \int_{0 < \Re u < 1} \frac{\Gamma(1-u) \Gamma^2(u)}{\Gamma(1+u)} z^u \frac{\total u}{2 \pi \mathi} \eqend{,}
\end{equation}
where the integration contour is a straight line parallel to the imaginary axis on which the real part $\Re u$ of the integration variable is fixed. We obtain
\begin{splitequation}
\label{eq:PerturbationOperator.MatrixElement.SingleIntegral}
I_1( s ) &\coloneqq \int_{-\infty}^\infty \ln\left( 1 + \mathe^{2 \pi v} \right) \frac{\mathe^{2 \pi \mathi v s}}{\cosh(\pi v)} \total v \\
&= \int_{0 < \Re u < \frac{1}{2}} \frac{\Gamma(1-u) \Gamma^2(u)}{\Gamma(1+u)} \int_{-\infty}^\infty \mathe^{2 \pi u v} \frac{\mathe^{2 \pi \mathi v s}}{\cosh(\pi v)} \total v \frac{\total u}{2 \pi \mathi} \\
&= \int_{0 < \Re u < \frac{1}{2}} \frac{\pi}{u \sin(\pi u)} \frac{1}{\cosh( \pi s - \mathi \pi u )} \frac{\total u}{2 \pi \mathi} \\
&= \frac{1}{2} \int_{-\infty}^\infty \frac{1}{(u_0+\mathi v) \sin(\pi u_0 + \mathi \pi v) \cosh( \pi s - \mathi \pi u_0 + \pi v )} \total v \eqend{,}
\end{splitequation}
where $u_0 \in (0,1/2)$, and where we used the reflection formula~\citeDLMFeq{5.5}{3} for the $\Gamma$ function to simplify the integrand. While the integral representation for the logarithm is valid for $0 < \Re u < 1$, we can exchange the $u$ and $v$ integrals by the Fubini--Tonelli theorem only in the range $0 < \Re u < \frac{1}{2}$, since only then the integrations converge absolutely. Closing the integration contour in the upper half plane, see figure~\ref{fig:MassiveContour-v-plane}, there are two infinite series of poles at $v^{\mathrm{s}}_k = \mathi u_0 + \mathi k$ and $v^\mathrm{c}_k = - s + \mathi (u_0+1/2) + \mathi k$, of which the one at $v^{\mathrm{s}}_0 = \mathi u_0$ is a second-order pole while all other poles are simple.
\begin{figure}
    \centering
    \includegraphics{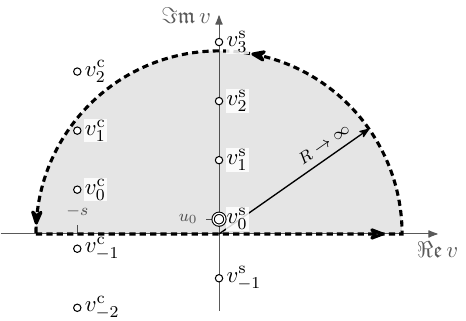}
    \caption{The integration contour of the integral~\eqref{eq:PerturbationOperator.MatrixElement.SingleIntegral} is closed in the upper half plane to give the dashed contour, where the contribution of the semicircle vanishes in the limit of infinite radius. The integrand has two sequences of simple poles $v^\mathrm{s}_k$, $v^\mathrm{c}_k$ from the $\sin$ and $\cosh$ function, respectively, while the pole at $v^\mathrm{s}_0$ has degree 2.}
    \label{fig:MassiveContour-v-plane}
\end{figure}
The contribution from the semicircle at infinity vanishes because of the exponential decay of the $\sin$ function in imaginary directions. Thus, summing the residues at these poles leads to 
\begin{splitequation}
\label{eq:PerturbationOperator.MatrixElement.SingleIntegral.Solution}
I_1( s ) &= \frac{1}{\cosh(\pi s)} \left[ \mathi \pi \tanh(\pi s) + \frac{2}{1 + 2 \mathi s} + \sum_{k=1}^\infty \left( \frac{2}{1 + 2 k + 2 \mathi s} - \frac{1}{k} \right) \right] \\
&= - \frac{1}{\cosh(\pi s)} \left[ \gamma_\mathrm{E} + \psi\left( \frac{1}{2} - \mathi s \right) \right] \eqend{,}
\end{splitequation}
where we employed the sum~\citeDLMFeq{5.7}{6} and the identities~\citeDLMFeq{5.4}{17} and~\citeDLMFeq{5.5}{2}. Here, $\gamma_\mathrm{E}$ is the Euler--Mascheroni constant, and $\psi(x) = \partial_x \ln \Gamma(x)$ is the digamma function.

Let us now turn to the double integral in~\eqref{eq:PerturbationOperator.MatrixElement.Integrating} to which we refer as $I_2( s, t )$. We change variables to $\sigma = \pi (v+w)$ and $\eta = \pi (v-w)$, such that
\begin{splitequation}
\label{eq:PerturbationOperator.MatrixElement.DoubleIntegral}
I_2( s, t ) &\coloneqq \iint_{-\infty}^\infty \ln\left[ \left( \mathe^{2 \pi v} - \mathe^{2 \pi w} \right)^2 \right] \frac{\mathe^{2 \pi \mathi v s - 2 \pi \mathi w t}}{\cosh(\pi v) \cosh(\pi w)} \total v \total w \\
&= \frac{1}{\pi^2} \iint_{-\infty}^\infty \ln\left[ 4 \mathe^{2 \sigma} \sinh^2(\eta) \right] \frac{\mathe^{\mathi (s-t) \sigma + \mathi (s+t) \eta}}{\cosh(\eta) + \cosh(\sigma)} \total \sigma \total \eta \\
&= \frac{1}{\pi^2} \int_{-\infty}^\infty \mathe^{\mathi (s+t) \eta} \left[ \ln\left( 4 \sinh^2 \eta \right) - 2 \mathi \partial_s \right] \int_{-\infty}^\infty \frac{\mathe^{\mathi (s-t) \sigma}}{\cosh(\eta) + \cosh(\sigma)} \total \sigma \total \eta \\
&= \frac{2}{\pi} \int_{-\infty}^\infty \mathe^{\mathi (s+t) \eta} \left[ \ln\left( 4 \sinh^2 \eta \right) - 2 \mathi \partial_s \right] \frac{\sin[ \eta (s-t) ]}{\sinh[ \pi (s-t) ] \sinh \eta} \total \eta \\
&= \frac{2}{\pi} \int_{-\infty}^\infty \mathe^{\mathi (s+t) \eta} \left[ \ln\left( 4 \sinh^2 \eta \right) - 2 \mathi \eta \cot[ \eta (s-t) ] + 2 \pi \mathi \coth[ \pi (s-t) ] \right] \\
&\qquad\qquad\times \frac{\sin[ \eta (s-t) ]}{\sinh[ \pi (s-t) ] \sinh \eta} \total \eta \\
&= 4 \mathi \frac{\cosh[ \pi (s-t) ]}{\sinh^2[ \pi (s-t) ]} \int_{-\infty}^\infty \mathe^{\mathi (s+t) \eta} \frac{\sin[ \eta (s-t) ]}{\sinh \eta} \total \eta \\
&\quad- \frac{4 \mathi}{\pi \sinh[ \pi (s-t) ]} \int_{-\infty}^\infty \mathe^{\mathi (s+t) \eta} \frac{\eta \cos[ \eta (s-t) ]}{\sinh \eta} \total \eta \\
&\quad+ \frac{\mathi}{\pi \sinh[ \pi (s-t) ]} \int_{-\infty}^\infty \ln\left( 4 \sinh^2 \eta \right) \frac{\mathe^{2 \mathi t \eta} - \mathe^{2 \mathi s \eta}}{\sinh \eta} \total \eta \eqend{.}
\end{splitequation}
For the first integral term we have 
\begin{splitequation}
\label{eq:PerturbationOperator.MatrixElement.DoubleIntegral.SinTerm.Simplifying}
I_2^{\sin}( s, t ) &\coloneqq \int_{-\infty}^\infty \mathe^{\mathi (s+t) \eta} \frac{\sin[ \eta (s-t) ]}{\sinh \eta} \total \eta \\
&= 2 \int_0^\infty \cos[ (s+t) \eta ] \frac{\sin[ \eta (s-t) ]}{\sinh \eta} \total \eta \\
&= \int_0^\infty \frac{\sin(2 s \eta)}{\sinh \eta} \total \eta - \int_0^\infty \frac{\sin(2 t \eta)}{\sinh \eta} \total \eta \eqend{,}
\end{splitequation}
and then we apply formula~\citeDLMFeq{4.40}{8} (analytically continued to $a \to 2 \mathi s$) to obtain
\begin{equation}
\label{eq:PerturbationOperator.MatrixElement.DoubleIntegral.SinTerm}
I_2^{\sin}( s, t ) = \frac{\pi}{2} \left[ \tanh(\pi s) - \tanh(\pi t) \right] \eqend{.}
\end{equation}
Analogously, the second integral term is 
\begin{splitequation}
\label{eq:PerturbationOperator.MatrixElement.DoubleIntegral.CosTerm}
I_2^{\cos}( s, t ) &\coloneqq \int_{-\infty}^\infty \mathe^{\mathi (s+t) \eta} \frac{\eta \cos[ \eta (s-t) ]}{\sinh \eta} \total \eta \\
&= \frac{1}{2} ( \partial_s + \partial_t ) \int_0^\infty \frac{\sin(2 s \eta) + \sin(2 t \eta)}{\sinh \eta} \total \eta \\
&= \frac{\pi}{4} ( \partial_s + \partial_t ) \left[ \tanh(\pi s) + \tanh(\pi t) \right] \\
&= \frac{\pi^2}{4} \left[ \cosh^{-2}(\pi s) + \cosh^{-2}(\pi t) \right] \eqend{.}
\end{splitequation}
For the third integral term, we introduce a small parameter $\epsilon$ 
\begin{equation}
I_2^{\ln,\epsilon}( s, t ) \coloneqq \int_{-\infty}^\infty \ln\left( 4 \sinh^2 \eta + \epsilon^2 \right) \frac{\mathe^{2 \mathi t \eta} - \mathe^{2 \mathi s \eta}}{\sinh \eta} \total \eta \eqend{,}
\end{equation}
such that the integral $I_2^{\ln,\epsilon}( s, t )$ coincides with the integral in the third term when taking the limit $\epsilon \to 0^+$, because the singularity at $\eta = 0$ is integrable. (Indeed, the fraction has a finite limit as $\eta \to 0^+$, and the logarithm for small $\eta$ has the approximate form $\ln(4 \eta^2) = 2 \ln(2 \eta)$ with indefinite integral $2 \eta [ \ln(2 \eta) - 1 ]$, which is finite as $\eta \to 0^+$). Once again, we use the Mellin--Barnes representation for the logarithm and obtain 
\begin{splitequation}
I_2^{\ln,\epsilon}( s, t ) &= 2 \ln \epsilon \int_{-\infty}^\infty \frac{\mathe^{2 \mathi t \eta} - \mathe^{2 \mathi s \eta}}{\sinh \eta} \total \eta \\
&+ \int_{0 < \Re u < \frac{1}{2}} \frac{\Gamma(1-u) \Gamma^2(u)}{\Gamma(1+u)} \left( \frac{4}{\epsilon^2} \right)^u \int_{-\infty}^\infty \left( \sinh^2 \eta \right)^u \frac{\mathe^{2 \mathi t \eta} - \mathe^{2 \mathi s \eta}}{\sinh \eta} \total \eta \frac{\total u}{2 \pi \mathi} \\
&= 2 \ln \epsilon \int_{-\infty}^\infty \frac{\mathe^{2 \mathi t \eta} - \mathe^{2 \mathi s \eta}}{\sinh \eta} \total \eta \\
&\quad+ 2 \pi \mathi \int_{0 < \Re u < \frac{1}{2}} \frac{1}{u \sin(\pi u)} \left( \frac{4}{\epsilon^2} \right)^u \int_0^\infty \left[ \sin(2 t \eta) - \sin(2 s \eta) \right] \sinh^{2u-1} \eta \total \eta \frac{\total u}{2 \pi \mathi} \eqend{.}
\end{splitequation}
Again, we restricted the integration to $0 < \Re u < \frac{1}{2}$ to be able to use Fubini--Tonelli and interchange the integrals (since only then the $\eta$ integral is absolutely convergent). To perform the integral over $\eta$ in the first term, we use the same computation as for the result~\eqref{eq:PerturbationOperator.MatrixElement.DoubleIntegral.SinTerm}, 
\begin{equation}
\int_{-\infty}^\infty \frac{\mathe^{2 \mathi t \eta}}{\sinh \eta} \total \eta = 2 \mathi \int_0^\infty \frac{\sin(2 t \eta)}{\sinh \eta} \total \eta = \mathi \pi \tanh( \pi s ) \eqend{.}
\end{equation}
For the other integral over $\eta$, we change variables to $\eta = 2 \artanh(x)$ and obtain 
\begin{splitequation}
\int_0^\infty \mathe^{2 \mathi t \eta} \sinh^{2u-1} \eta \total \eta &= 4^u \int_0^1 \frac{x^{2 u - 1} (1-x)^{- 2 u - 2 \mathi t}}{(1+x)^{2 u - 2 \mathi t}} \total x \\
&= 4^u \frac{\Gamma(2 u) \Gamma(1 - 2 \mathi t - 2 u)}{\Gamma(1 - 2 \mathi t)} \hypergeom{2}{1}\left( 2 u - 2 \mathi t, 2 u; 1 - 2 \mathi t; - 1 \right) \\
&= 4^u \frac{\Gamma(2 u) \Gamma(1 - 2 \mathi t - 2 u) \Gamma(u - \mathi t + 1)}{\Gamma(2 u - 2 \mathi t + 1) \Gamma(-u - \mathi t + 1)} \\
&= 4^{-u} \Gamma\left( \frac{1}{2} - u - \mathi t \right) \Gamma\left( \frac{1}{2} - u + \mathi t \right) \frac{\cos[ \pi (u - \mathi t) ]}{\sin(2 \pi u) \Gamma(1-2u)} \eqend{,}
\end{splitequation}
where we used the integral representation~\citeDLMFeq{15.6}{1} of the Gauss hypergeometric function $\hypergeom{2}{1}$ (there written in terms of the Olver hypergeometic function~\citeDLMFeq{15.2}{2}). To obtain the last two lines, we use the special value~\citeDLMFeq{15.4}{26} and the well-known relations for the $\Gamma$ function~\cite[\href{https://dlmf.nist.gov/5.5}{Sec.~5.5}]{DLMF}. (All restrictions on parameters are fulfilled for $0 < u < \frac{1}{2}$, and the final result holds for $0 < \Re u < \frac{1}{2}$ by analyticity.) Finally, we need to integrate over $u$ and take the limit $\epsilon \to 0^+$. In the limit, the first term is logarithmically divergent, and the second term is divergent as well since $\Re u > 0$. To show that their sum is actually finite, we need to move the integration contour to have $\Re u < 0$. This we do by lifting it over the pole at $u = 0$, see figure~\ref{fig:MassiveContour-u-plane}, 
\begin{figure}
    \centering
    \includegraphics{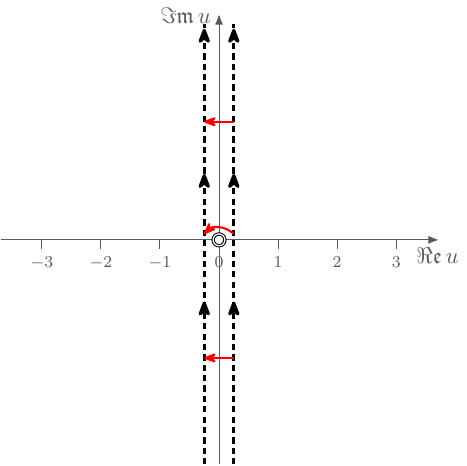}
    \caption{When moving the integration contour (dashed) over the double pole at $u = 0$ (red arrows), we have to include the residue of the pole.}
    \label{fig:MassiveContour-u-plane}
\end{figure}
according to the schematic
\begin{equation}
\int_{\Re u > 0} f(u) \frac{\total u}{2 \pi \mathi} = \Res_{u = 0} f(u) + \int_{\Re u < 0} f(u) \frac{\total u}{2 \pi \mathi} \eqend{.}
\end{equation}
The remaining integral with $\Re u < 0$ is uniformly bounded by a term $\epsilon^{- 2 \Re u}$, and vanishes as $\epsilon \to 0^+$, such that only the residue term remains: 
\begin{splitequation}
I_2^{\ln,\epsilon}( s, t ) &= 2 \pi \mathi \left[ \tanh(\pi t) - \tanh(\pi s) \right] \ln \epsilon \\
&\quad+ \Res_{u = 0} \frac{2 \pi \mathi \epsilon^{-2u}}{u \sin(2 \pi u) \Gamma(1-2 u)}
\Biggl[ \sinh(\pi t) \Gamma\left( \frac{1}{2} - u - \mathi t \right) \Gamma\left( \frac{1}{2} - u + \mathi t \right) \\
&\hspace{12em}- \sinh(\pi s) \Gamma\left( \frac{1}{2} - u - \mathi s \right) \Gamma\left( \frac{1}{2} - u + \mathi s \right) \Biggr] \eqend{.}
\end{splitequation}
Computing the residue, the divergence in $\ln \epsilon$ cancels since the pole at $u = 0$ is of second order, such that a logarithm appears. We introduce a shorter notation for the sum of two digamma functions and obtain the result
\begin{equations}
\tilde{\psi}(s) &\coloneqq \psi\left( \frac{1}{2} + \mathi s \right) + \psi\left( \frac{1}{2} - \mathi s \right) \eqend{,} \\
\lim_{\epsilon \to 0} I_2^{\ln,\epsilon}( s, t ) &= \mathi \pi \left[ 2 \gamma_\mathrm{E} + \tilde{\psi}(s) \right] \tanh(\pi s) - \mathi \pi \left[ 2 \gamma_\mathrm{E} + \tilde{\psi}(t) \right] \tanh(\pi t) \eqend{.}
\end{equations}
Collecting all integration results, we have $K(s,t)$ as given in~\eqref{eq:PerturbationOperator.MatrixElement.Solution}.

Finally, we compute the integral kernel of the modular operator~\eqref{eq:ModularHamiltonian.Interval.Massive.Integral}. To perform the integral, it is useful to pass to sum and difference variables $\sigma = (s+t)/2$ and $\eta = (s-t)/2$. For a shorter notation of the ratio that appears repeatedly, let us define
\begin{equation}
r(x,y) \coloneqq \sqrt{\frac{\ell - x}{\ell + x} \frac{\ell - y}{\ell + y}} > 0 \eqend{.}
\end{equation}
There are three integral contributions to the modular Hamiltonian 
\begin{equation}
H^{(1)}_{12}(x,y) = \frac{2 \mathi m \ell^2}{\sqrt{ (\ell^2-x^2) (\ell^2-y^2) }} \left[ 4 \ln\left( m \ell \mathe^{2\gamma_\mathrm{E}} \right) \vartheta_1(x,y) + \vartheta_2(x,y) + \vartheta_3(x,y) \right] \eqend{,}
\end{equation}
which are
\begin{equations}
\vartheta_1(x,y) &\coloneqq \int r^{2 \mathi \sigma}(x,y) \total \sigma \int r^{2 \mathi \eta}(x,-y) \frac{\eta}{\sinh(2 \pi \eta)} \total \eta \eqend{,} \\
\vartheta_2(x,y) &\coloneqq \iint r^{2 \mathi \sigma}(x,y) r^{2 \mathi \eta}(x,-y) \frac{\eta}{\sinh(2 \pi \eta)} \left[ \tilde{\psi}( \sigma - \eta ) + \tilde{\psi}( \sigma + \eta ) \right] \total \sigma \total \eta \eqend{,} \\
\vartheta_3(x,y) &\coloneqq \iint r^{2 \mathi \sigma}(x,y) r^{2 \mathi \eta}(x,-y) \frac{\eta}{\sinh(2 \pi \sigma)} \left[ \tilde{\psi}( \sigma - \eta ) - \tilde{\psi}( \sigma + \eta ) \right] \total \sigma \total \eta \eqend{.}
\end{equations}
Note that the integrals only converge in a distributional sense once again. We use the integral representation~\citeDLMFeq{5.9}{16} for the digamma function
\begin{equation}
\psi(z) = - \gamma_\mathrm{E} - \lim_{\epsilon \to 0^+} \left[ \int_0^{1-\epsilon} \frac{t^{z-1}}{1-t} \total t + \ln \epsilon \right] \eqend{.}
\end{equation}
For $\vartheta_{2,3}$, we pull the limits out of the converging integrals, and write 
\begin{equation}
\vartheta_{i}(x,y) = \lim_{\epsilon \to 0^+} \vartheta_{i,\epsilon}(x,y) \eqend{,} \quad i \in \{2,3\} \eqend{.}
\end{equation}
Then all integrals are of the form (for some $z > 0$)
\begin{equations}
\int_{-\infty}^{\infty} z^{2 \mathi s} \total s &= 2 \pi \, \delta(2 \ln z) \eqend{,} \\
\int_{-\infty}^{\infty} z^{2 \mathi s} s \total s &= - 2 \pi \mathi \, \delta'(2 \ln z) \eqend{,} \\
\int_{-\infty}^{\infty} z^{2 \mathi s} \frac{1}{\sinh(2 \pi s)} \total s &= - \frac{\mathi}{2} \frac{1-z}{1+z} \eqend{,} \\
\int_{-\infty}^{\infty} z^{2 \mathi s} \frac{s}{\sinh(2 \pi s)} \total s &= \frac{z}{2 (1+z)^2} \eqend{.}
\end{equations}
Since the variable $z$ depends on $x$ and $y$, in the last step, we use the composition formula for the Dirac $\delta$ distribution~\eqref{eq:Delta.Composition} and its derivative~\eqref{eq:DeltaDeriv.Composition}. Applying these steps to the integrals, we obtain 
\begin{equation}
\vartheta_1(x,y) = \pi \frac{r(x,-y)}{\left[ 1 + r(x,-y) \right]^2} \delta\left[ 2 \ln r(x,y) \right] = \frac{\pi}{8 \ell^3} ( \ell^2 - x^2 )^2 \delta(x+y) \eqend{,}
\end{equation}
\begin{splitequation}
\vartheta_{2,\epsilon}(x,y) &= - 4 ( \gamma_\mathrm{E} + \ln \epsilon ) \vartheta_1(x,y) \\
&\quad- \pi \int_0^{1-\epsilon} \left[ \frac{\sqrt{t} r(x,-y)}{[ 1 + \sqrt{t} r(x,-y) ]^2} + \frac{\sqrt{t} r(x,-y)}{[ \sqrt{t} + r(x,-y) ]^2} \right] \\
&\qquad\quad\times \Big[ \delta\left[ \ln t + 2 \ln r(x,y) \right] + \delta\left[ \ln t - 2 \ln r(x,y) \right] \Big] \frac{t^{-\frac{1}{2}}}{1-t} \total t \\
&= - 4 ( \gamma_\mathrm{E} + \ln \epsilon ) \vartheta_1(x,y) \\
&\quad- \pi \frac{2 \ell^2 - x^2 - y^2}{4 \ell^2} \int_0^{1-\epsilon} \left[ \delta\left[ t - r^2(x,y) \right] + \delta\left[ t - r^2(-x,-y) \right] \right] \frac{\sqrt{t}}{1-t} \total t \eqend{,}
\end{splitequation}
and
\begin{splitequation}
\vartheta_{3,\epsilon}(x,y) &= \pi \int_0^{1-\epsilon} \left[ \frac{\sqrt{t} - r(x,y)}{\sqrt{t} + r(x,y)} - \frac{1 - \sqrt{t} r(x,y)}{1 + \sqrt{t} r(x,y)} \right] \\
&\qquad\quad\times \Big[ \delta'\left[ \ln t + 2 \ln r(x,-y) \right] + \delta'\left[ \ln t - 2 \ln r(x,-y) \right] \Big] \frac{t^{-\frac{1}{2}}}{1-t} \total t \\
&= \frac{\pi (x-y) \sqrt{ (\ell^2-x^2) (\ell^2-y^2) }}{8 \ell^3} \int_0^{1-\epsilon} \left[ \delta\left[ t - r^2(x,-y) \right] - \delta\left[ t - r^2(-x,y) \right] \right] \total t \\
&\quad- \frac{\pi \sqrt{ (\ell^2-x^2) (\ell^2-y^2) }}{2 \ell^2} \int_0^{1-\epsilon} \Big[ r^2(x,-y) \, \delta'\left[ t - r^2(x,-y) \right] \\
&\hspace{15em}+ r^2(-x,y) \, \delta'\left[ t - r^2(-x,y) \right] \Big] \total t \eqend{,} \raisetag{1.8em}
\end{splitequation}
where we also used~\eqref{eq:DeltaDeriv.Product} to simplify the integrands, and combined some terms. To perform the integrals over $t$, we introduce a Heaviside $\Theta$ distribution as 
\begin{equation}
\int_0^{1-\epsilon} f(t) \total t = \int_0^\infty \Theta(1-\epsilon-t) f(t) \total t \eqend{,}
\end{equation}
and then use the Dirac $\delta$ distributions to perform the $t$ integrals explicitly. Collecting all integral terms $\vartheta_i$, we obtain
\begin{splitequation}
\label{eq:MassiveModularHamiltonian.Correction.Appendix}
H^{(1)}_{12}(x,y) &= 2 \pi \mathi m \ell \lim_{\epsilon \to 0^+} \Bigg[ \ln\left( \frac{m \ell}{\epsilon} \mathe^{\gamma_\mathrm{E}} \right) \frac{\ell^2-x^2}{2 \ell^2} \delta(x+y) \\
&\qquad- \frac{2 \ell^2 - x^2 - y^2}{8 \ell^2 (x+y)} \left[ \Theta\left( \frac{2 \ell (x+y)}{(\ell+x) (\ell+y)} - \epsilon \right) - \Theta\left( - \frac{2 \ell (x+y)}{(\ell-x) (\ell-y)} - \epsilon \right) \right] \\
&\qquad+ \frac{x-y}{8 \ell^2} \left[ \Theta\left( \frac{2 \ell (x-y)}{(\ell+x) (\ell-y)} - \epsilon \right) - \Theta\left( - \frac{2 \ell (x-y)}{(\ell-x) (\ell+y)} - \epsilon \right) \right] \\
&\qquad- \frac{r^2(x,-y)}{2 \ell} \delta\left( \frac{2 \ell (x-y)}{(\ell+x) (\ell-y)} - \epsilon \right) - \frac{r^2(-x,y)}{2 \ell} \delta\left( - \frac{2 \ell (x-y)}{(\ell-x) (\ell+y)} - \epsilon \right) \Bigg] \eqend{.}
\end{splitequation}
Finally, we have to take the limit $\epsilon \to 0^+$. In the Dirac $\delta$ distributions, we simply set $\epsilon = 0$ and use the composition formula~\eqref{eq:Delta.Composition}, since the resulting expression is a well-defined distribution.

For the Heaviside $\Theta$ distributions in the second line, this is not so simple, since the term that they multiply potentially diverges. To compute the limit, we change to sum and difference variables $\sigma = (x+y)/(2\ell)$ and $\eta = (x-y)/(2\ell)$ and obtain for $x,y \in [-\ell,\ell]$
\begin{splitequation}
\frac{2 \ell^2 - x^2 - y^2}{\ell (x+y)} \Theta\left( \frac{2 \ell (x+y)}{(\ell+x) (\ell+y)} - \epsilon \right) &= \frac{2 \ell^2 - x^2 - y^2}{\ell (x+y)} \Theta\Big( 2 \ell (x+y) - (\ell+x) (\ell+y) \epsilon \Big) \\
&= \frac{1 - \sigma^2 - \eta^2}{\sigma} \Theta\left( 4 \sigma - \left[ (1+\sigma)^2 - \eta^2 \right] \epsilon \right) \eqend{.}
\end{splitequation}
Using the result~\eqref{eq:PfThetaCoordinate.Limit.Appendix} with the replacements $x \to \sigma$, $a \to \eta$ and $\epsilon \to \epsilon/4$, we can take the limit $\epsilon \to 0^+$ and obtain
\begin{splitequation}
&\lim_{\epsilon \to 0^+} \left[ \frac{1 - \sigma^2 - \eta^2}{\sigma} \Theta\left( 4 \sigma - \left[ (1+\sigma)^2 - \eta^2 \right] \epsilon \right) + (1-\eta^2) \ln\left[ \mu \frac{1-\eta^2}{4} \epsilon \, \mathe^{-\gamma_\mathrm{E}} \right] \delta(\sigma) \right] \\
&\quad= (1-\eta^2) \pvalue_\mu \frac{\Theta(\sigma)}{\sigma} - \sigma \, \Theta(\sigma) \eqend{.}
\end{splitequation}
Switching back to $x$ and $y$ and using the almost homogeneous scaling~\eqref{eq:PfTheta.Scaling.Appendix}, we thus obtain
\begin{splitequation}
&\lim_{\epsilon \to 0^+} \Bigg[ \frac{2 \ell^2 - x^2 - y^2}{\ell (x+y)} \Theta\left( \frac{2 \ell (x+y)}{(\ell+x) (\ell+y)} - \epsilon \right) \\
&\hspace{6em}+ \frac{4 \ell^2 - (x-y)^2}{2 \ell} \ln\left( \mu \frac{4 \ell^2 - (x-y)^2}{8 \ell} \epsilon \right) \delta(x+y) \Bigg] \\
&\quad= \frac{4 \ell^2 - (x-y)^2}{2 \ell} \pvalue_\mu \frac{\Theta(x+y)}{x+y} + 2 \frac{\ell^2 - x^2}{\ell} \gamma_\mathrm{E} \, \delta(x+y) - \frac{x+y}{2 \ell} \Theta(x+y) \eqend{.}
\end{splitequation}

The analogous computation, using the limit~\eqref{eq:PfThetaCoordinate.Limit2.Appendix}, establishes that
\begin{splitequation}
&\lim_{\epsilon \to 0^+} \Bigg[ \frac{2 \ell^2 - x^2 - y^2}{\ell (x+y)} \Theta\left( - \frac{2 \ell (x+y)}{(\ell-x) (\ell-y)} - \epsilon \right) \\
&\hspace{6em}- \frac{4 \ell^2 - (x-y)^2}{2 \ell} \ln\left( \mu \frac{4 \ell^2 - (x-y)^2}{8 \ell} \epsilon \right) \delta(x+y) \Bigg] \\
&\quad= \frac{4 \ell^2 - (x-y)^2}{2 \ell} \pvalue_\mu \frac{\Theta(-(x+y))}{x+y} - \frac{4 \ell^2 - (x-y)^2}{2 \ell} \gamma_\mathrm{E} \, \delta(x+y) - \frac{x+y}{2 \ell} \Theta(-(x+y)) \eqend{,}
\end{splitequation}
such that our result~\eqref{eq:MassiveModularHamiltonian.Correction.Appendix} reduces to
\begin{splitequation}
H^{(1)}_{12}(x,y) &= 2 \pi \mathi m \ell \Bigg[ \ln\left( m \ell \frac{\ell^2 - x^2}{2 \ell} \mu \right) \frac{\ell^2 - x^2}{2 \ell^2} \delta(x+y) + \frac{\abs{x-y}}{8 \ell^2} - \frac{\ell^2 - x^2}{2 \ell^2} \delta(x-y) \\
&\quad- \frac{4 \ell^2 - (x-y)^2}{16 \ell^2} \pvalue_\mu \frac{\Theta(x+y)}{x+y} + \frac{4 \ell^2 - (x-y)^2}{16 \ell^2} \pvalue_\mu \frac{\Theta(-(x+y))}{x+y} + \frac{\abs{x+y}}{16 \ell^2} \Bigg] \eqend{.} \raisetag{3em}
\end{splitequation}
Finally, using the relation~\eqref{eq:PfTheta.Relation.Appendix} and the fact that
\begin{equation}
(x+y)^2 \pvalue_\mu \frac{1}{\abs{x+y}} = \abs{x+y} \eqend{,}
\end{equation}
which follows straightforwardly from the definition~\eqref{eq:ModifiedPrincipalValueDistribution}, we obtain the result~\eqref{eq:FirstOrderMassCorrection.Result}.

\bibliography{references}
\addcontentsline{toc}{section}{References}

\end{document}